\documentclass[a4paper,10pt,reqno]{article}  

\usepackage[top=1.2cm, bottom=1.2cm, left=2cm, right=3cm]{geometry}

\usepackage{amscd,caption}
\usepackage{lipsum}
\usepackage[breaklinks,colorlinks]{hyperref}
\usepackage{breakcites}
\usepackage[arrow,matrix]{xy}
\usepackage{graphicx}
\usepackage{amsmath, amsfonts, amsthm, comment,mathtools}
\usepackage[T1]{fontenc}
\usepackage{color}
\usepackage{soul}
\usepackage{chngcntr}
\usepackage{authblk,blkarray}
\usepackage{marvosym}
\usepackage[title]{appendix}
\usepackage{setspace}
\usepackage{subcaption}
\usepackage{titlesec}

\titleformat{\subsection}
  {\normalfont\fontsize{11}{15}\bfseries}{\thesubsection}{1em}{}

\newtheorem{rem}{Remark}
\newtheorem{cor}{Corollary}
\newtheorem{prop}{Proposition}
\newtheorem{lem}{Lemma}

\DeclareMathOperator*{\argmax}{arg\,max}
\DeclareMathOperator*{\argmin}{arg\,min}
\DeclareMathOperator{\diag}{diag} 

\newcommand{\R}{\mathbb{R}}

\newcommand{\PP}{\mathbb{P}}

\newcommand{\E}{\mathbb{E}}

\newcommand{\x}{\mathbf{x}}

\newcommand{\f}{\mathbf{f}}

\newcommand{\p}{\mathbf{p}}

\newcommand{\bel}[1]{\begin{equation}\label{#1}}

\newcommand{\be}{\begin{equation}}
\newcommand{\ba}{\begin{eqnarray}}
\newcommand{\ea}{\end{eqnarray}}

\newcommand{\qe}{\end{equation}}
\newcommand{\al}{\alpha}

\newcommand{\ld}{\lambda}
\newcommand{\supp}{\rm{supp}}
\newcommand{\de}{\delta}

\newcommand{\De}{\Delta}

\newcommand{\suml}{\sum\limits}

\newcommand{\eps}{\varepsilon}

\DeclareMathOperator*{\var}{Var}
\DeclareMathOperator*{\cov}{Cov}
\DeclareMathOperator*{\corr}{Cor}

\usepackage{titlesec}
\titlelabel{\thetitle.\:}

\begin{document}


\title{[Accepted Manuscript]\\
\vspace*{.5cm}
Adaptive Bet-Hedging Revisited: Considerations of Risk and Time Horizon}
\author[a,c]{Omri Tal\footnote{{Corresponding author: omrit1248@gmail.com }
}
}
\author[a,b]{Tat Dat Tran\footnote{trandat@mis.mpg.de \& tran@math-uni.leipzig.de} 
}

\affil[a]{\small Max Planck Institute for Mathematics in the Sciences, Inselstrasse 22, D-04103 Leipzig, Germany} 

\affil[b]{\small Institute of Mathematics, Leipzig University, Augustusplatz 10, D-04109 Leipzig, Germany}

\affil[c]{\small Cohn Institute for the History and Philosophy of Science of Ideas, Tel Aviv University}

\maketitle

\begin{abstract}
{\normalsize
\noindent
Models of adaptive bet-hedging commonly adopt insights from Kelly's famous work on optimal gambling strategies and the financial value of information. In particular, such models seek evolutionary solutions that maximize long term average growth rate of lineages, even in the face of highly stochastic growth trajectories. Here, we argue for extensive departures from the standard approach to better account for evolutionary contingencies. Crucially, we incorporate considerations of volatility minimization, motivated by interim extinction risk in finite populations, within a finite time horizon approach to growth maximization. We find that a game-theoretic competitive-optimality approach best captures these additional constraints, and derive the equilibria solutions under straightforward fitness payoff functions and extinction risks. We show that for both maximal growth and minimal time relative payoffs the log-optimal strategy is a unique pure-strategy symmetric equilibrium, invariant with evolutionary time horizon and robust to low extinction risks.

}
\end{abstract}

\noindent {\em Keywords: adaptive bet-hedging, Kelly gambling, growth optimal portfolio theory, game theory, finite time horizon, extinction risk.}



\section{Introduction}
{\normalsize
\begin{verbatim}
``Adversity has the effect of eliciting talents, which in prosperous circumstances 
would have lain dormant.'' -- Horace (65BC-8BC)
\end{verbatim}
}

\noindent
Kelly's work on optimal gambling strategies and the value of side information was arguably the first convincing attempt at applying concepts from information theory for analysis in a different field \cite{Kelly56}. This work was the precursor to growth-optimal portfolio theory which has extended the basic ideas to the realm of capital markets (\cite{ct06}). There has recently been a resurge of interest in employing insights from optimal gambling theory in models of adaptive bet-hedging under fluctuating environments, where close analogies between the economic and biological setting have been convincingly made apparent (\cite{bergstrom14}; \cite{rl11}; \cite{donaldson10}). 

Biological bet hedging was originally proposed to explain the observation of un-germinated seeds of annual plants (\cite{cohen66}). This strategy involves the variable phenotypic expression of a single genotype, rather than a result of genetic polymorphism, although it is difficult to empirically determine whether observed phenotypic diversity in a population arises from randomization by identical genomes or from an underlying polymorphism (\cite{sb87}). Indeed, evolutionary biologists have long acknowledged that in a stochastically variable environment, natural selection is likely to favor a gene that randomizes its phenotypic expression (\cite{bergstrom14}). Recent work has revealed a variety of potential instances of bet hedging populations: delayed germination in desert winter annual plants that meets postulated criteria of adaptive bet hedging in a variable environment (\cite{gv14}),  bacterial persistence in the presence of antibiotics that appears to constitute an adaptation tuned to the distribution of environmental change (\cite{kussell05}), flowering times in Lobelia inflata which point to flowering being a conservative bet-hedging strategy (\cite{sj03}), or even bet-hedging as a behavioural phenotype, such as the case of nut hoarding in squirrel populations in anticipation of short or long winters (\cite{bergstrom14}).

Notwithstanding these empirical findings, identifying actual cases of adaptive bet hedging in the wild remains elusive. As \cite{sb87} have noted more than three decades ago, it is in general difficult to determine whether observed diversity of behavior in a population arises from randomization by genetically identical individuals or from genetic heterogeneity within co-located individuals optimized for different environmental conditions. Moreover, phenotypic heterogeneity can arise within genetically homogenous populations as a form of specialization in a stable environment through stochastic gene expression, positive feedback loops, or asymmetrical cell division, all processes where bet-hedging is not at play (\cite{rd17}). These difficulties provide further impetus for constructing better and more elaborate models to test against the data.

Of particular note in classic bet hedging models is the adoption from economic theory of asymptotic growth rate optimality as the target function for fitness maximization strategies, where growth in wealth is analogous to growth in lineage size. Indeed, since evolution proceeds by shifting gene frequencies over generations, with frequency changes being multiplicative, long-term fitness is commonly measured by geometric mean fitness across generations (\cite{hopper18}). At the same time, it is also widely acknowledged that long-run growth rate is not a valid measure of fitness under fluctuating environments, such as in the case of bet-hedging populations (\cite{lande07}). 

The resulting intrinsic unpredictability has led some researchers to formulate a probabilistic perspective for natural selection that integrates various effects of uncertainty on natural selection (\cite{yoshimura09}). The applicability of geometric mean fitness has also come into question under finite population models, where the probability of fixation provides additional and sometimes more suitable information than the geometric mean fitness (\cite{pd01}), and in periodically cycling selection regimes, where evolutionary success depends on the length of the cycle and the strength of selection (\cite{Ram18}). Moreover, both gambling and bet-hedging models targeting optimal growth rate implicitly assume an infinite time horizon in formulating the geometric average, and thereby ignore the finiteness of actual horizons over which both economic and evolutionary processes ultimately act. The problem is further amplified when interim extinction risk is taken into account, especially under finite population models. Lineage growth trajectories which are highly stochastic are at risk of large `draw-downs', which may pull the population below some extinction threshold, despite possessing a high asymptotic growth rate. Here we aim to incorporate considerations of finite evolutionary horizons and extinction risk in the search for adaptive optimality in bet hedging models.

\subsection{Background: The standard model}

Most adaptive bet hedging models are largely based on the classic horse-race gambling model associated with Kelly (1956), where the biological counter-part is a lineage apportioning bets on several possible environments. Assume that $k$ horses run in a race, and let horse $X_i$ win with probability $p_i$. If horse $X_i$ wins, the odds are $o_i$ for $1$. A gambler wishes to apportion his bankroll among the horses $0<f_i\le 1$, such that $\sum f_i=1$ and participate in indefinitely repeated races $n\to \infty$. How to best apportion the bankroll each time? In this setting, wealth is a discrete-time stochastic process over $n$ periods, 
\[
W_n=\prod_{i=1}^n W_i(X)
\] 
where $W(X)=f(X)O(X)$ is the random factor by which the gambler's wealth is multiplied when horse X wins. More explicitly,
\[
W_n (f)=\prod_{i=1}^k (f_i o_i )^{H_i}, \quad \text{where, } H\sim Multinomial(n,k,[p_1,...,p_k]).
\]
Kelly first insight was that choosing to simply maximize expected wealth (for any time horizon $n$) gives $\argmax_f E[W_n (f)]=1$, with the implication that one bets everything on a single horse (the one with the highest $p_i$) and a consequent chance of total ruin once that horse loses a race. Therefore, Kelly proposed maximizing the asymptotic growth rate (the rigorous justification provided by \cite{breiman61}). By the law of large numbers random wealth may be expressed as,
\[
W_n (f) \doteq 2^{n E[\log W(X)]}
\]
where,
\[
E[\log W(X)]=\sum_{i=1}^k p_i \log  f_i o_i 
\]
is the asymptotic exponential growth rate. If the gambler stakes his entire wealth each time i.e., $\sum f_i=1$, then
\[
E[\log W(X)]=\sum_{i=1}^k p_i \log o_i -H(p)-D(p||f)
\]
is maximized (convex nonlinear optimization) at ``proportional gambling'' $f=p$ where $D(p\|f)$ is minimized, without regard the actual odds provided by the bookie.

Indeed, the notion of {\it proportional gambling}, made famous by Kelly's treatment, has found its way into classic models of diversified bet hedging. In such models often assumed that ``appropriate phenotypes are produced in proportion to the likelihood of each environment'' (\cite{hopper18}) and that consequently ``the classical bet-hedging prediction [is] that the optimum probability for employing a strategy is approximately equal to the probability that the strategy will be useful'' (\cite{km07}). Here we follow recent approaches that extend the standard model to non-lethal environments via a full fitness matrix, such that this notion is no longer directly applicable.  

\cite{breiman61} was first to show that the Kelly solution is optimal in two convincing ways: [a] that given a Kelly strategy $\phi^*$ and any other ``essentially different'' strategy $\phi$ (not necessarily a fixed fractional betting strategy),
\[
\lim_{n\to \infty}  \frac{W_n (\phi^*)}{W_n (\phi)}= \infty\quad a.s
\]
and [b] that it minimizes the expected time to reach asymptotically large wealth goals. Moreover, this strategy is myopic in the sense that at each iteration of the race one only needs to consider the presently given parameters (\cite{hakansson71}). 
However, Kelly strategies may also yield tremendous drawdowns, a problem widely recognized in the gambling community, such that optimal Kelly is often viewed as ``too risky''; in practice gamblers and investors use `fractional Kelly' which deviates from the optimal solution but reduces the effective variance of the stochastic growth (Fig.~\ref{fig:1}). In the biological framework, this can lead to abrupt extinction events in finite (especially small) populations with highly stochastic lineage growth trajectories. A further complication is that the underlying probability distributions are merely estimated from past data and model assumptions, leading often to over-betting and increased risk (\cite{maclean10}). 
\begin{figure}[h]
\centering\includegraphics[width=8cm]{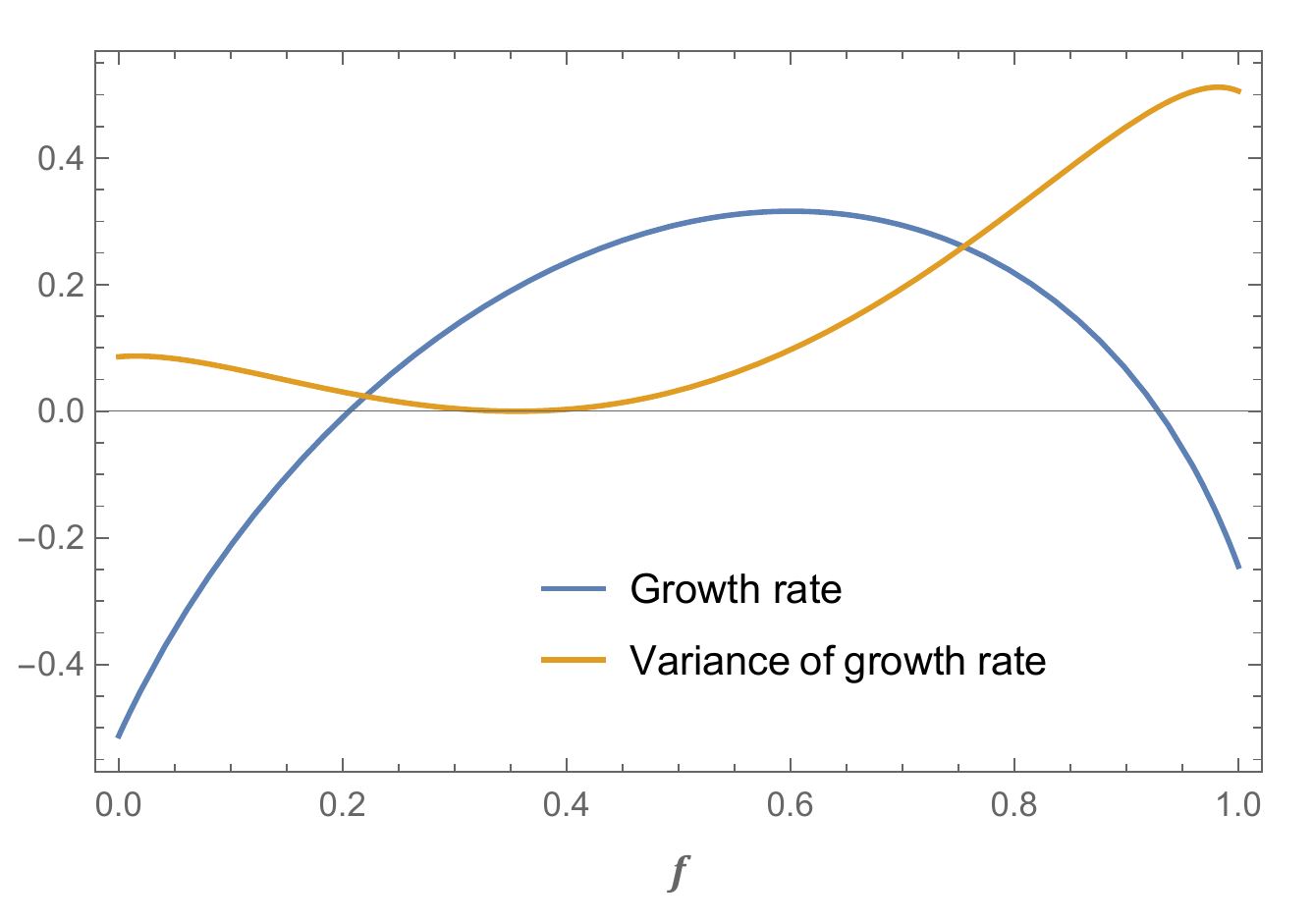}
\caption{{\small The asymptotic exponential growth rate and its (finite-horizon) variance for either the two-horse racing model or the classic bet-hedging model (two environments). Note that the strategy ($f$ on the $x$-axis) that maximizes the growth rate is far from the locus of minimal variance.}}
\label{fig:1}
\end{figure}\\



In this work, we extend the existing models to incorporate both interim extinction risk and finite evolutionary time horizons within a bet hedging framework. This requires re-conceptualizing geometric-mean fitness for such highly stochastic growth scenarios. We ultimately derive fitness functions that better account for such conditions where the fluctuating environment is strongly coupled to both long and short-term growth, and locate optimal stable equilibria.


\section{Methods}

\subsection{The full-fitness matrix model}

\noindent
We assume environments are i.i.d random events across generations, multinomially distributed (with some results generalized to non-identically distributed environments). Individuals within lineages have a static full fitness matrix $[O_{ij}]$ in which nonlethal environments have low but generally non-zero fitness (\cite{donaldson10}; \cite{rl11}). We adopt a finite-population model where lineages start off with some initial population size $W_0$, implicitly assumed higher than some bet-hedging evolutionary threshold (\cite{km07}). Lineages then evolve strategies to randomize individual phenotypes towards maximizing growth across finite horizons in the face of interim extinction threats. 
More formally, with $k$ environments and phenotypes,
\bel{eq:1}
[o_{ij}] :=  
\begin{blockarray}{cccc}
& t_1 & \cdots & t_k \\
\begin{block}{c(ccc)}
  e_1 & o_{11} & \cdots & o_{1k} \\
  \vdots & \vdots & \ddots & \vdots \\
  e_k & o_{k1} & \cdots & o_{kk} \\
  \end{block}
\end{blockarray}
\qe
the general model of lineage growth trajectory across $n$ generations under strategy $f$ is a random process,
\bel{eq:2}
W_n=\prod_{i=1}^k \Big(\sum_{j=1}^k f_j o_{ij}  \Big)^{H_i} 
\qe
 where,
\[
H\sim Multinomial(n,k,[p_1,\ldots,p_k ])
\]
with off-diagonal values reflecting the lower fitness for non-matching environments,
\[
o_{ii}>o_{ij} \ge 0 \text{ and } o_{ii}>1
\]
and where all individuals in a lineage are bet-hedging,
\[
\sum_{i=1}^k f_i =1.
\]

And finally, using a straightforward formulation of the growth rate, $W_n^{1/n}$, a random variable for any finite horizon.

We first derive the asymptotic growth-rate optimal ``Kelly'' solution for this setting ($f^{Kelly}$) with a corresponding bet-hedging region of the environment simplex (Appendix~\ref{app:A}). Relaxing the assumption of i.i.d environments we derive the static Kelly solution for the case of nonstationary environments -- where environments are independent but not identically distributed across generations (Appendix~\ref{app:B}). While under nonstationary environments an optimal growth rate is reached with a dynamic myopic strategy, we focus here on a static strategy since adaptations effectively stabilize across time spans much higher than single generations, such that from evolutionary considerations dynamic strategies are not likely to emerge. Alternative models of fluctuating environments such as Markov chains with underlying switching probabilities (e.g. \cite{xiangyi17}) are not pursued here and left for future work. Finally, we identify a `reference' strategy that admits deterministic growth trajectories, namely the ``Dutch book'' solution (where the variance of the finite-time growth rate is zero) and characterize the consequent loss of growth incurred by exchanging opportunity for certainty (Appendix~\ref{app:C}).

\subsection{Relative fitness payoff function}

\noindent
We now wish to go beyond the standard approach of targeting the optimization of the asymptotic growth rate as undertaken in the previous section -- to incorporate finite evolutionary horizons and extinction risk considerations. For the sake of simplicity, we confine our model here to the case of $k=2$ environments and phenotypes (so that the two environments occur with probability $p$ and $1-p$). To motivate the shift to a finite horizon framework we first highlight an important property of our stochastic growth model, known also in portfolio theory (\cite{markowitz06}). We prove that for any two essentially different strategies, the maximal time $n_0$ one lineage ``dominates'' the other is finite for every realization of lineage trajectory pair (Appendix~\ref{app:D}). The exponentially diminishing histogram of last intersection times of given two growth strategies in Fig.~\ref{fig:2}B (with a single instance of two trajectories for illustration in Fig.~\ref{fig:2}A) demonstrates this phenomenon. \\

\hspace*{.5cm}{\bf\large A} \hspace*{6.5cm} {\bf\large B}
\begin{figure}[htp!]
\centering\includegraphics[width=7cm]{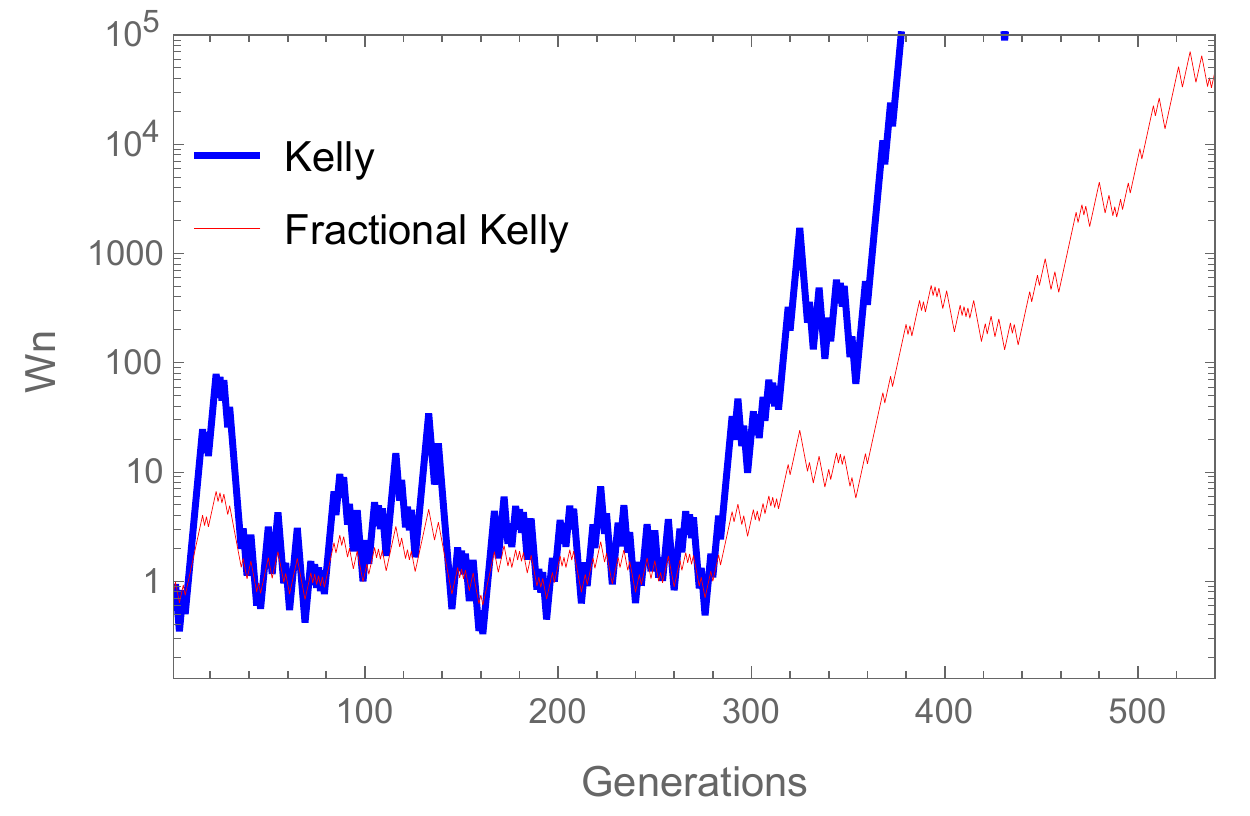}
\centering\includegraphics[width=7cm]{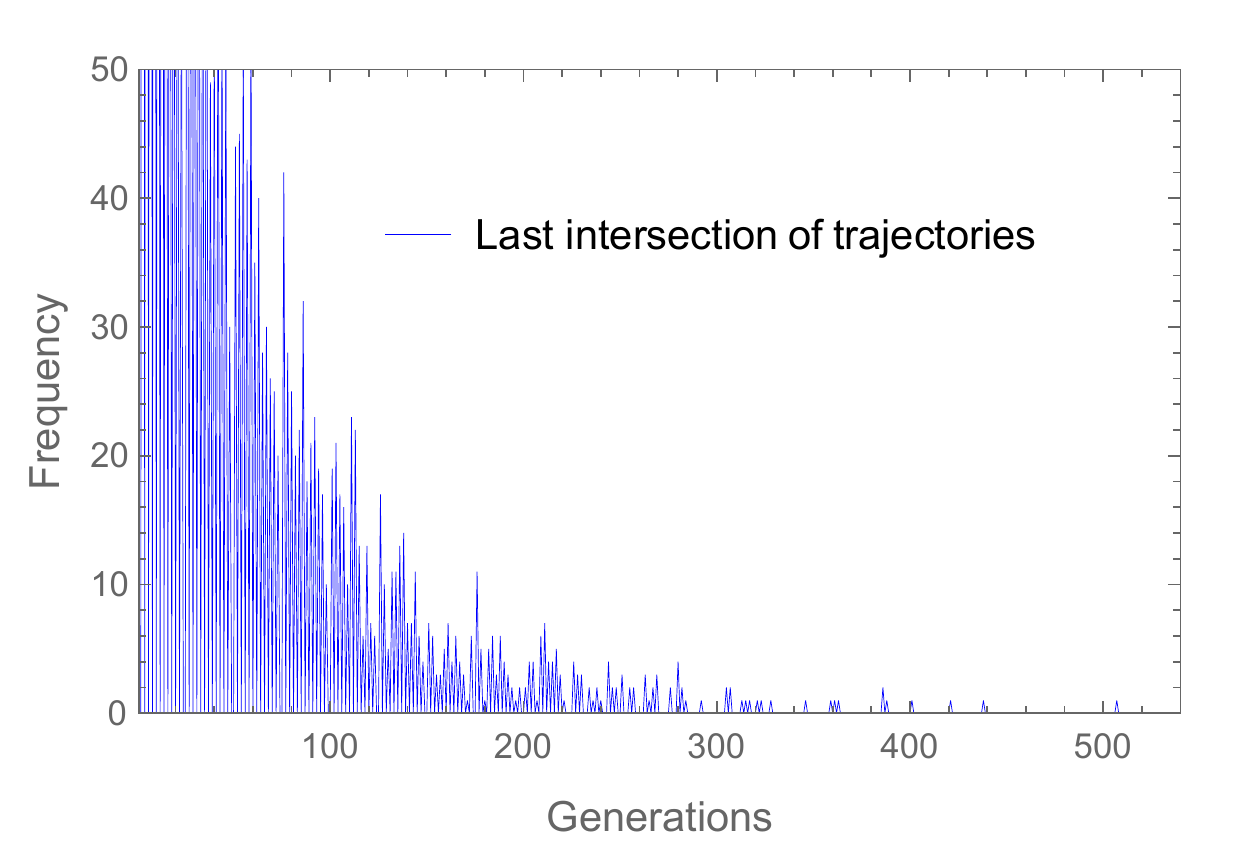}
\caption{{\small Two perspectives on growth trajectories $W_n$ and their last `intersection' point for two lineages with different strategies (a log-optimal Kelly and a suboptimal `fractional Kelly'). A: The optimal Kelly strategy eventually departs from a suboptimal strategy without further intersections, `the last intersection'. | B: A histogram demonstrating that the last intersection of the trajectories of any two growth strategies occurs at some {\it finite} time (here, somewhat above 500 generations).}}
\label{fig:2}
\end{figure}\\  
\vspace{.5cm}

The sustained variance and high skewness of the growth rate distribution under any finite horizon necessitates a comparative approach in formulating a fitness payoff function (in fact, the growth rate is asymptotically log-normal as shown in Appendix~\ref{app:E}). Consider a relative fitness measure for two different lineage strategies $f$ and $g$: the probability that a random trajectory of a lineage with strategy $f$ exceeds the random trajectory of a lineage with strategy $g$ (given time horizon $n$),
\bel{eq:3}
h(f,g)=P(W_n (f)>W_n (g)) 						
\qe
with an induced relation defined by, 
\bel{eq:4}
W_n (f)\ge W_n (g):  P(W_n (f)>W_n (g)) \ge P(W_n (f)<W_n (g)).                                
\qe

We may interpret this probabilistic relation between two strategies as relative fitness. Note that since realizations of $W_n (f)$ and $W_n (g)$ stem from the same underlying stochastic environmental sequence, they will generally be highly correlated (with the corresponding logarithmic growth rates in fact perfectly correlated, as shown in appendix~\ref{app:F}). Consequently, the probability in Eq.~\eqref{eq:3} must be derived from their joint distribution rather than simply from marginal distributions. Fig.~\ref{fig:3} depicts realizations of the log growth rates of $W_n (f)$ and $W_n (g)$ as histogram distributions for some choice of strategies $f$ and $g$, and some finite evolutionary horizon $n$. Asymptotically with time horizon $n$, such distributions approach normality with variance going to zero (Appendix~\ref{app:E}).
\begin{figure}[h]
\centering\includegraphics[width=8cm]{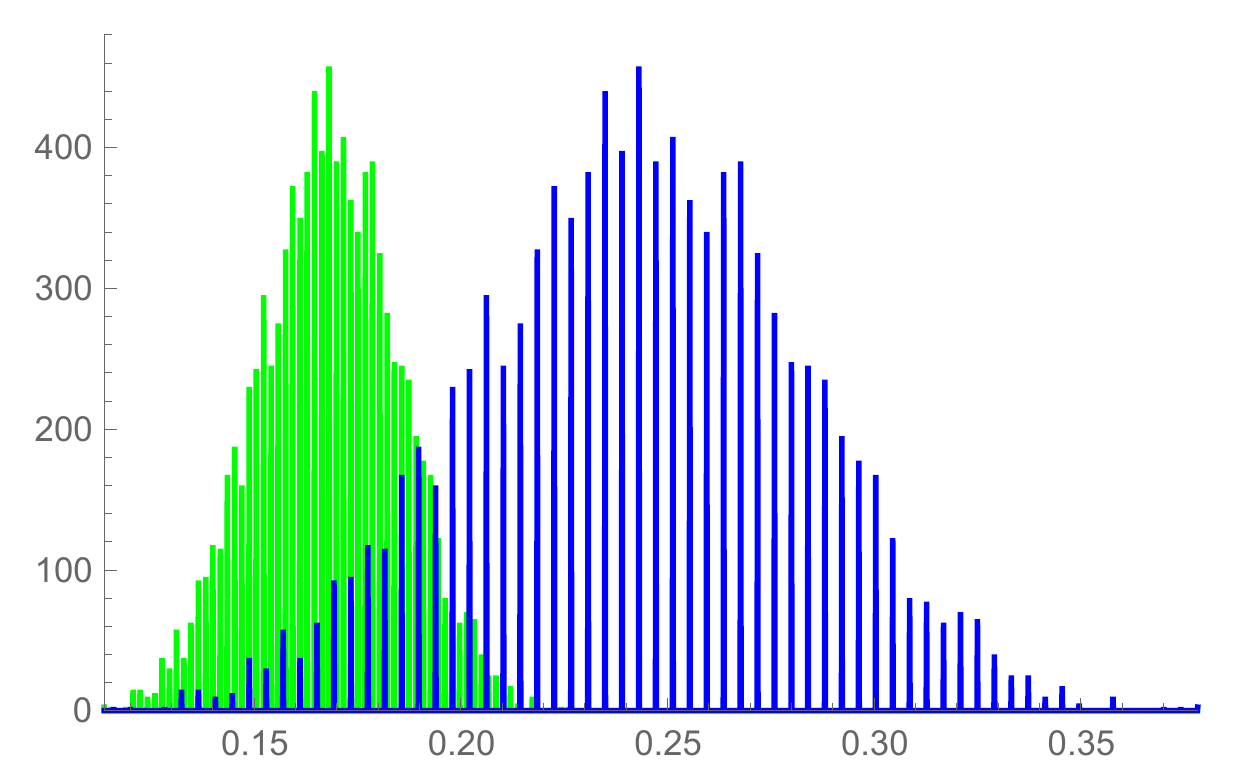}
\caption{{\small The distributions of log growth rates for two competing strategies here illustrated as histograms (the $x$-axis is the log growth rate). In this simple case shown, the strategy generating the higher growth-rates right distribution (blue) is evidently of higher fitness than the one generating the left distribution (green). For clarity, only the marginal distributions are shown (the complete picture is in the joint distribution, due to the correlation between any two growth trajectories that share the same fluctuating environment regime).}}
\label{fig:3}
\end{figure}\\
\vspace{.5cm}

A few properties of the order induced by this relation are worth highlighting. [a] it is a complete order since any two $W_n$ are comparable under the relation, [b] it is transitive for any $n$ and consequently a pre-order, and [c] its maximal element is $W_n^* (f^{Kelly})$, such that both the order induced by $E[\log W_n (f)]$ and the order induced by the payoff $P(W_n (f)>W_n (g))$ form complete preorders and have the same maximal element (Appendix~\ref{app:H}). Despite these beneficial properties, given any `vanilla' strategy $g$ and time horizon $n$, the strategy that maximizes the payoff function, 
\[\argmax_f P(W_n (f)>W_n (g))\]  
will vary as a function of $g$ and $n$ (demonstrated by counterexamples), and in particular will not necessarily be $f^{Kelly}$. This implies that a wildtype lineage with strategy $g$ different from $f^{Kelly}$ will eventually be overtaken by some mutant invasive lineage with a strategy that maximizes this payoff function, a process that may potentially remain in recurrent flux, with invasive lineages replacing a wildtype lineage.

\subsection{Competitive optimality with risk}

\noindent
To see whether evolutionary stable optima may also emerge we develop a game-theoretic approach. Players are lineages with particular bet-hedging strategies and random initial population size. Lineages interact by competing over a common niche subject to the same environmental fluctuations. This set-up is in some contrast to more standard evolutionary game theory settings, where agents are organisms rather than lineages and where the notion of an iterated strategy is prominent, but maintains the central aspect of interactions formalized in a payoff function (e.g., \cite{sn18}). A lineage survives the competitive encounter by avoiding extinction (defined in what follows) while exceeding its opponent in size over a given time horizon. This outcome is determined by a game-theoretic deterministic payoff function, modified from Eq.~\eqref{eq:3} to incorporate an extinction threshold and randomized initial lineage size. Ultimately, we are searching for Nash equilibria.

This approach is motivated by the classic work on time-invariant game-theoretic competitive optimality, within the scope of growth-optimal portfolio theory (\cite{bc80, bc88}). Bell and Cover consider a competitive setting for a stock portfolio model under any finite number of investment periods and prove that for any relative wealth payoff $E[\phi(UW_1/VW_2)]$ and portfolio wealth $W_1$ and $W_2$, there are conditions on the function $\phi$ such that the log-optimal Kelly portfolio is a solution to the game, given initial randomizations $U$ and $V$ (independent and of equal expectation). In particular, $\phi (x)=\chi_{[1,\infty)}(x)$ results in the payoff $\PP(UW_1\le VW_2)$ with the log-optimal portfolio as a game-theoretic solution, given some initial fair randomizations. This additional fair randomization reduces the effect of small differences in end wealth, thus avoiding unwanted cases where the optimal strategy is beat by a small amount most of the time (\cite{ct06}).

\subsection{The payoff function in a game-theoretic setting}

\noindent
For any time-horizon $n$ and extinction threshold $d$, we define a (deterministic) payoff function: the probability that a random trajectory of a lineage with strategy $f$ exceeds the random trajectory of a lineage with strategy $g$ without first going extinct (given time horizon $n$),   
\bel{eq:5}
\begin{split}
M_n (f,g)&=P(u_0 W_n (f)>v_0 W_n (g) | \text{ extinction level } d) :=\\
& P(u_0 W_n (f)>v_0 W_n (g) \wedge W_i,V_i>d, i:1,\ldots,n)+\\
&                                             P(u_0 W_n (f)>d, i:1,\ldots,n \wedge v_0 W_i (f)\le d \text{ some } i)                                          
\end{split}
\qe
with initial population size independent randomizations $u_0$ and $v_0$, independent and of same mean but possibly of a different distribution class. 

This payoff function induces a symmetric discrete-valued non constant-sum game setting, although it is conceptually ``zero-sum'' $M_n(f,g)+M_n(g,f)<1$ (Appendix~\ref{app:J}). Crucially, our payoff matrix is finite since it reflects the finitely many strategies possible in a finite population model -- there can only be $N$ different sized partitions of a population of size $N$ in betting on two environments (under $k=2$ environments and phenotypes). A low-resolution toy-model instance of the payoff matrix is depicted in Fig.~\ref{fig:4}. 
\begin{figure}[h]
\centering\includegraphics[width=6cm]{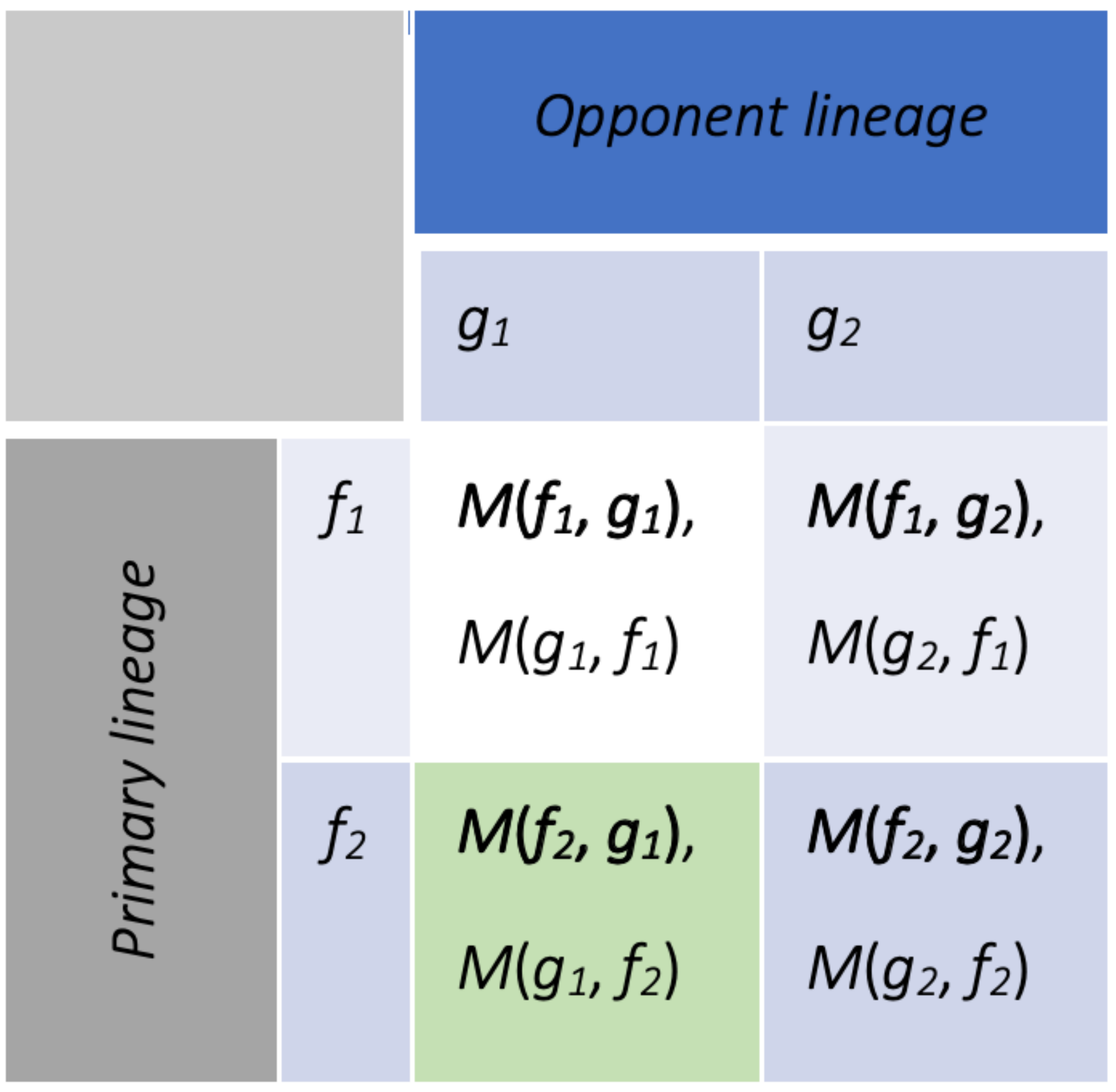}
\caption{{\small An example of a ($2 \times 2$ strategies) evolutionary payoff matrix for a game of two lineages, with primary lineage payoff in bold.}}
\label{fig:4}
\end{figure}
\vspace{.5cm}

Our goal would be to identify pure strategy Nash equilibria reflecting the evolutionary solutions to competitive bet-hedging. In particular, we would like to explore the conditions under which a bet-hedging setting admits a symmetric equilibrium and whether it is unique. In Appendix~\ref{app:K} we prove that for an infinite-size payoff matrix (i.e., continuous strategies) the log-optimal strategy is the solution to this game, invariant with the choice of time horizon. Moreover, any finite matrix representing the $N$ strategies possible for a lineage of finite size $N$ necessarily also admits a solution, as illustrated in Fig.~\ref{fig:5}. This solution is the strategy closest to the log-optimal strategy under the finite resolution framework, such that it converges to it asymptotically with $N$ (Appendix~\ref{app:N}). Finally, under a nonstationary environment model the log-optimal strategy again emerges as the equilibrium static strategy -- even given short time horizons (Appendix~\ref{app:O}). \\

\noindent
\hspace*{.5cm}{\bf\large A} \hspace*{3.5cm} {\bf\large B}\hspace*{3cm} {\bf\large C}\hspace*{3.2cm} {\bf\large D}
\begin{figure}[h]
\centering\includegraphics[width=3.5cm]{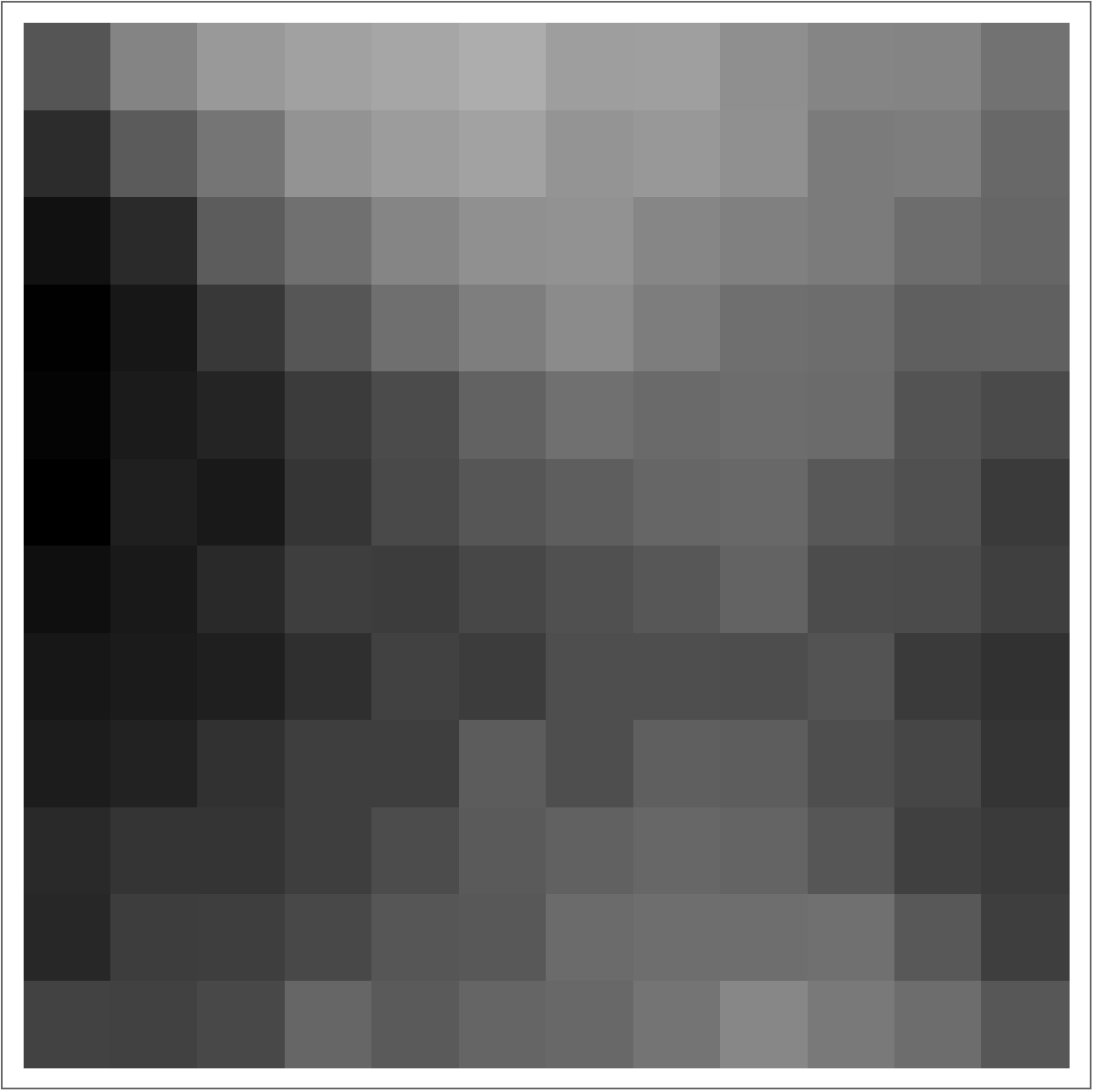}
\centering\includegraphics[width=3.5cm]{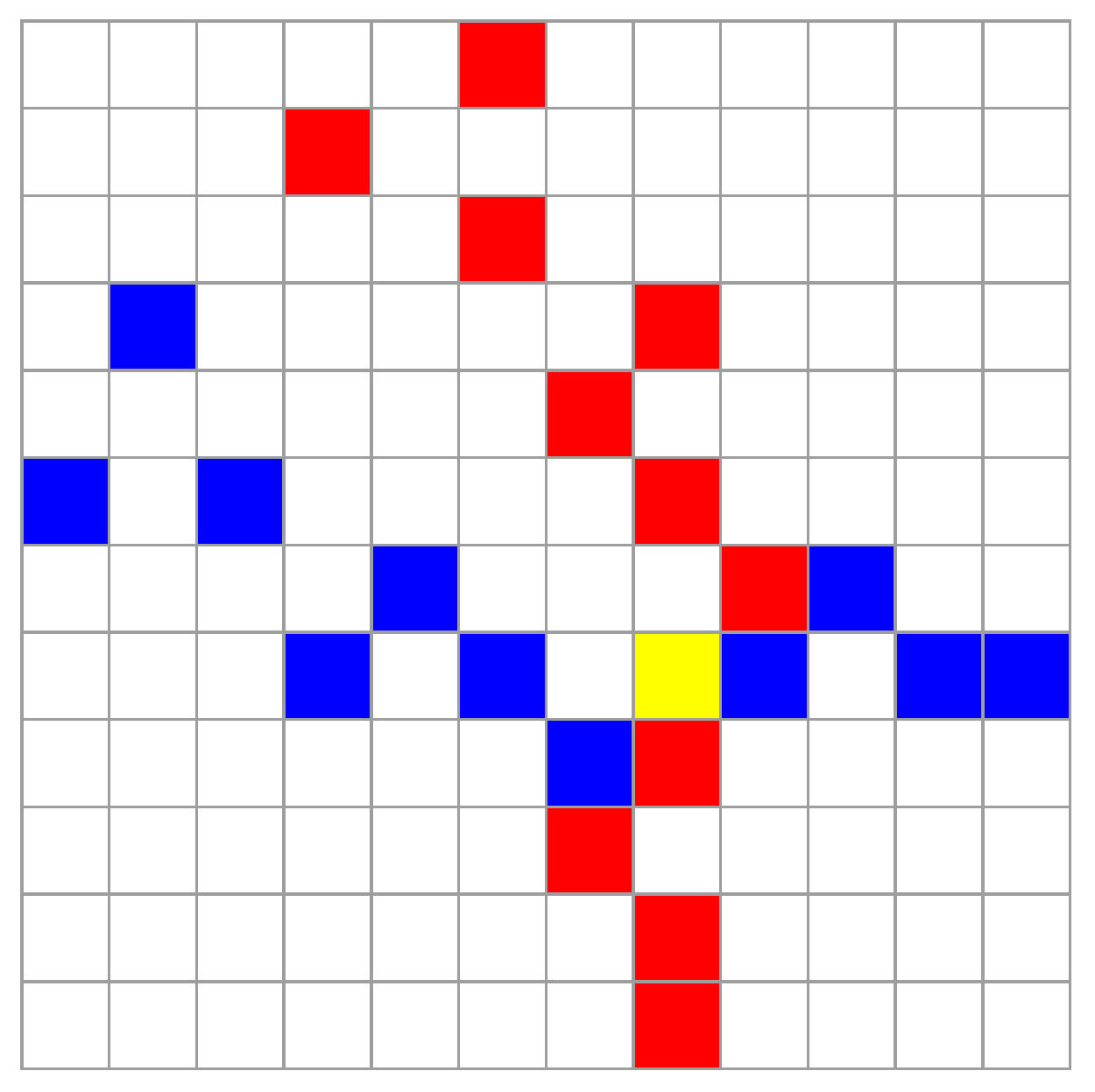}
\centering\includegraphics[width=3.5cm]{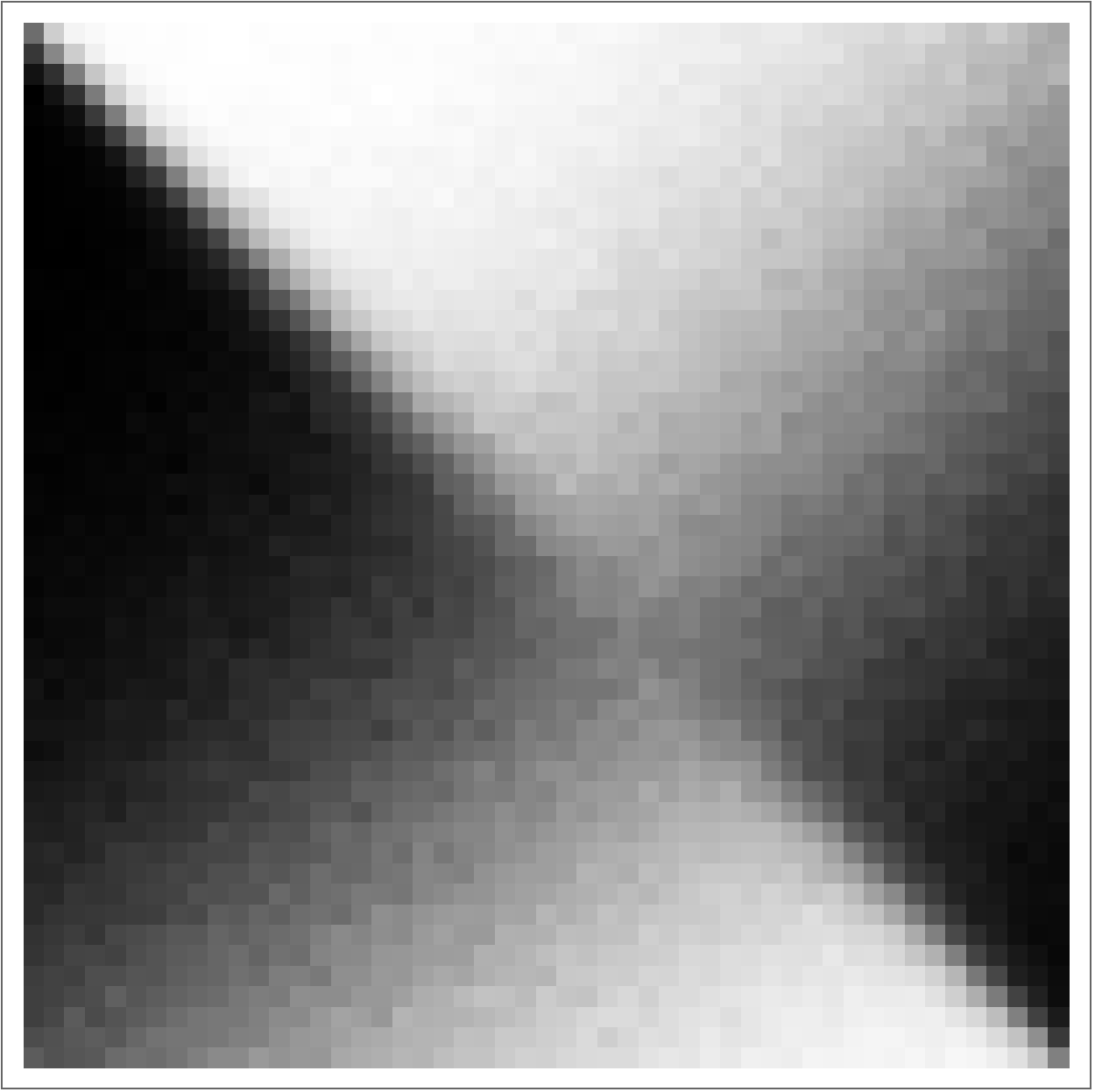}
\centering\includegraphics[width=3.5cm]{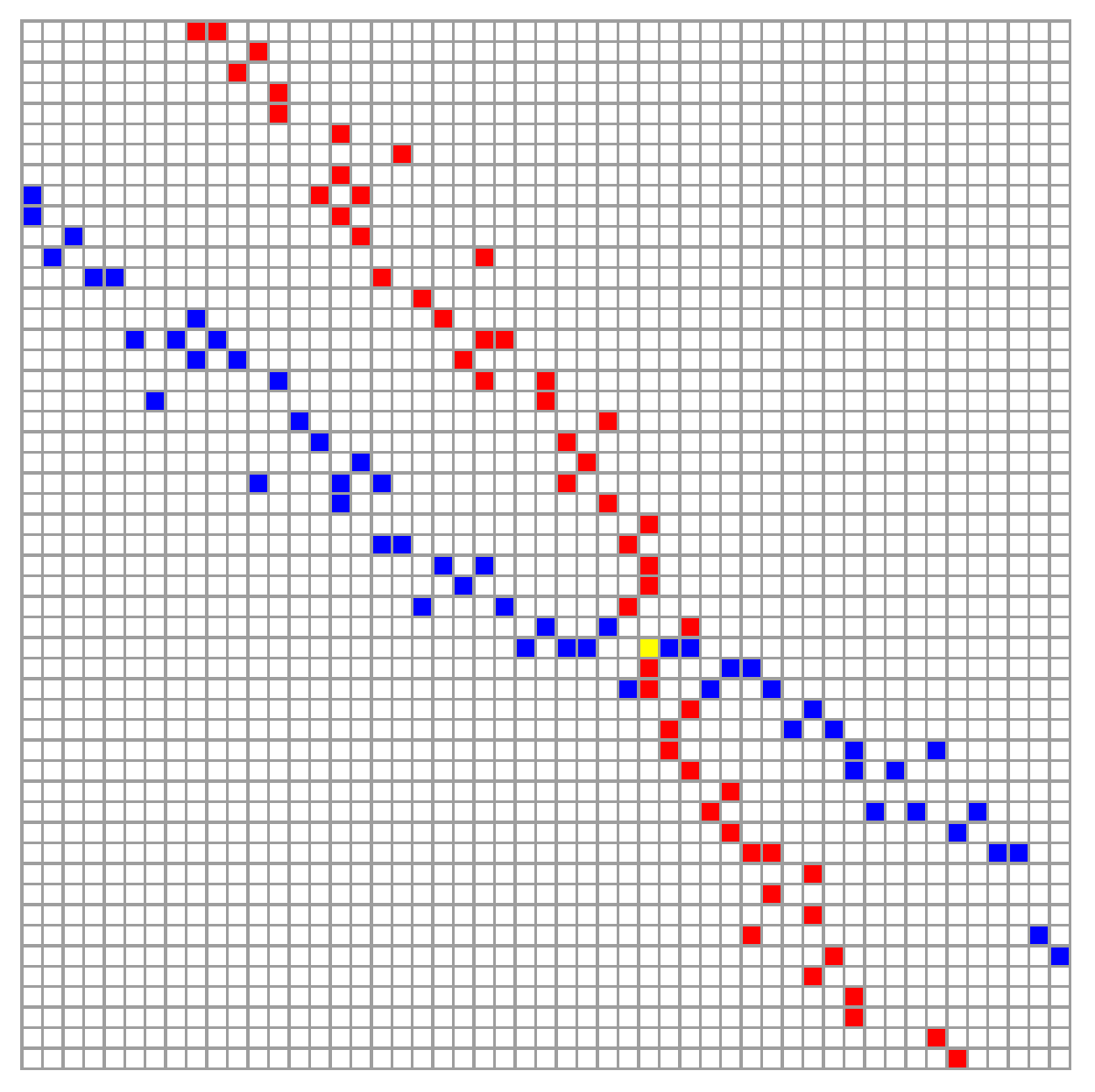}
\caption{{\small An example of payoff matrix simulations (figs. A and C: darker cells represent higher probabilities) and the resulting Nash equilibria in the maximal element matrices (figs. B and D: blue cells for column-maximal primary player payoffs, red for row-maximal opponent payoffs and yellow for symmetric Nash equilibria). Only portions of the matrices around the equilibrium are shown. A+B: Low resolution matrices corresponding to a small population with limited strategies | C+D: Higher resolution matrices, which correspond to a larger population and subsequently higher number of strategies (multiple maxima in some adjacent cells of panel D is an effect due to a combination of using finite runs in the simulation of the payoff function, such that computed probabilities are rational values, along with high resolution in the range of strategies). The model uses a $2\times 2$ fitness matrix: $[o_{11}=3.0;o_{12}=0.2;o_{22}=1.8;o_{21}=0.1]$ with $p=0.594$ and a resulting log-optimal strategy $f=0.60$).}}
\label{fig:5}
\end{figure}\\

The effect of lineage-size extinction thresholds on actual rates of extinction of random growth trajectories is illustrated in Fig.~\ref{fig:6}A. As would be expected, higher thresholds of extinction correspond to higher probabilities of extinction, with extinction rates that converge quickly to asymptotic values (Appendix~\ref{app:P}). Numerical simulations indicate that when incorporating low extinction thresholds that result in low extinction rates, the symmetric Nash equilibrium remains stable at the log-optimal strategy. Higher thresholds may result in a number of scenarios: a shift of the symmetric equilibrium away from the log-optimal solution, complete lack of equilibrium solution, or the emergence of multiple symmetric equilibria; in conjunction, multiple pairs of off-diagonal equilibria may appear (see Fig.~\ref{fig:6}B for one such scenario).\\

\vspace*{.7cm}
\hspace*{1.5cm}{\bf\large A} \hspace*{7.5cm} {\bf\large B}
\begin{figure}[h]
\centering\includegraphics[width=9cm]{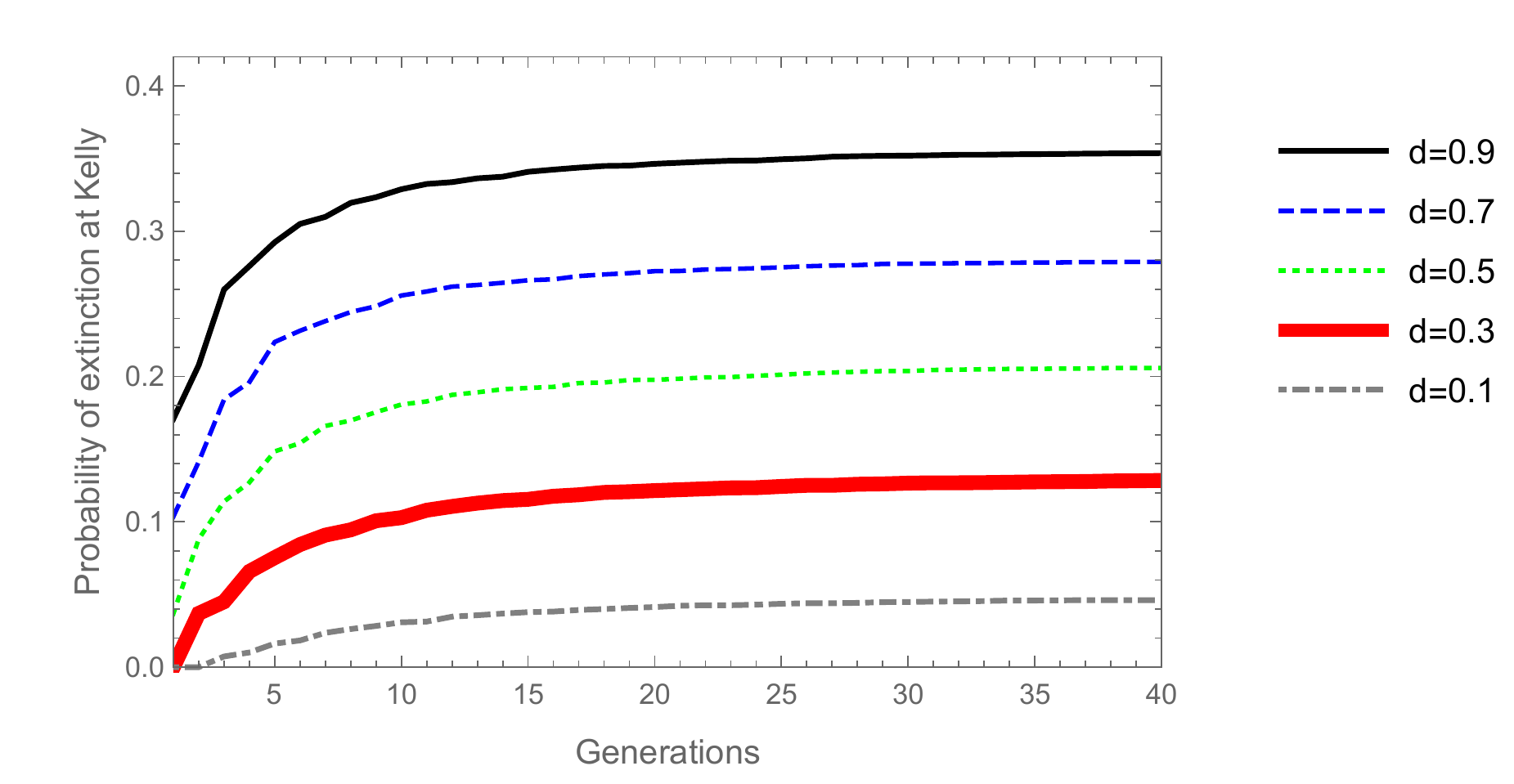}
\centering\includegraphics[width=4.5cm]{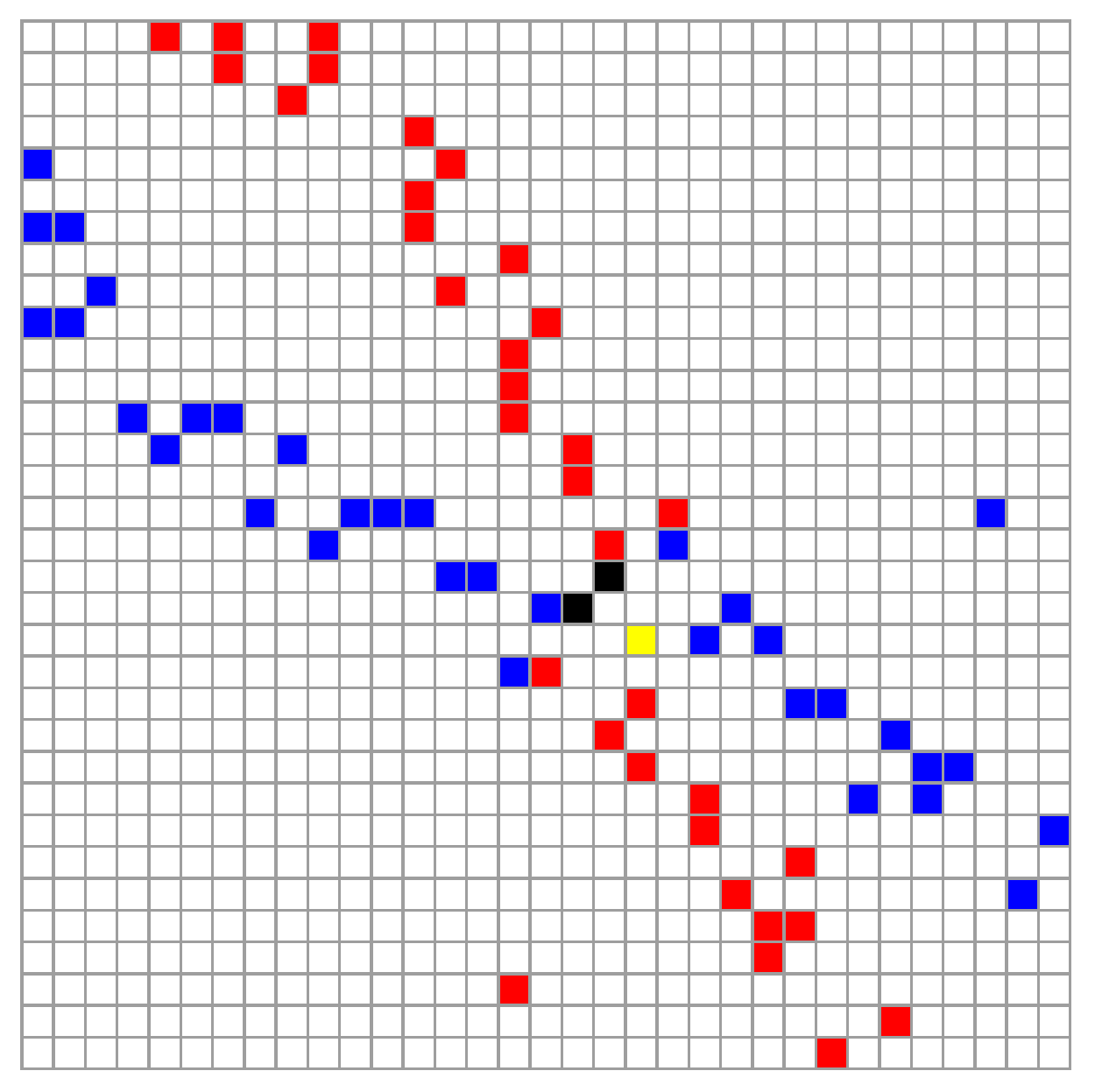}
\caption{{\small A: The accumulative probability of extinction for log-optimal (Kelly) strategies for different lineage-size extinction thresholds. | B: An instance of the maximal element matrix of a payoff matrix (a portion around the solutions) resulting from introducing an extinction threshold: two off-diagonal Nash equilibria (black) along with a symmetric Nash equilibrium shifted from the log-optimal strategy (yellow). The simulations use a $2 \times 2$ fitness matrix: $[o_{11}=2.5;o_{12}=0.2;o_{22}=1.4;o_{21}=0.1]$ with $p=0.75$, and with $n=60$ in the simulation of panel B.}}
\label{fig:6}
\end{figure}\\

\subsection{Minimum time to reach a population threshold size}

\noindent
To gain further perspective on optimal strategies under highly stochastic growth we consider evolutionary competition between lineages, where survival is determined by reaching a certain threshold of lineage size in minimal time (e.g. for $K-$selected species, see \cite{reznick02}). In effect, the lineage with growth characteristics that minimize the time to reach a certain population size threshold ``wins'', a setting with potential relevance in the context of competitively colonizing a limited niche, as in range expansion scenarios (see \cite{martine19} for a bet-hedging population expanding into an unoccupied space). We follow the classic results of \cite{breiman61} on the log-optimal portfolio as the optimal strategy minimizing the expected time to reach an asymptotic target wealth, but instead of an infinite target we base the fitness payoff function on finite targets. Initial insight into the effect of strategy choice on the consequent distributions of minimal time (Fig.~\ref{fig:7}A) is provided by comparing their expectation, where the optimality of Kelly is already apparent (Fig.~\ref{fig:7}B).\\

\noindent
\hspace*{1cm}{\bf\large A} \hspace*{6.5cm} {\bf\large B}
\begin{figure}[h]
\centering\includegraphics[width=7cm]{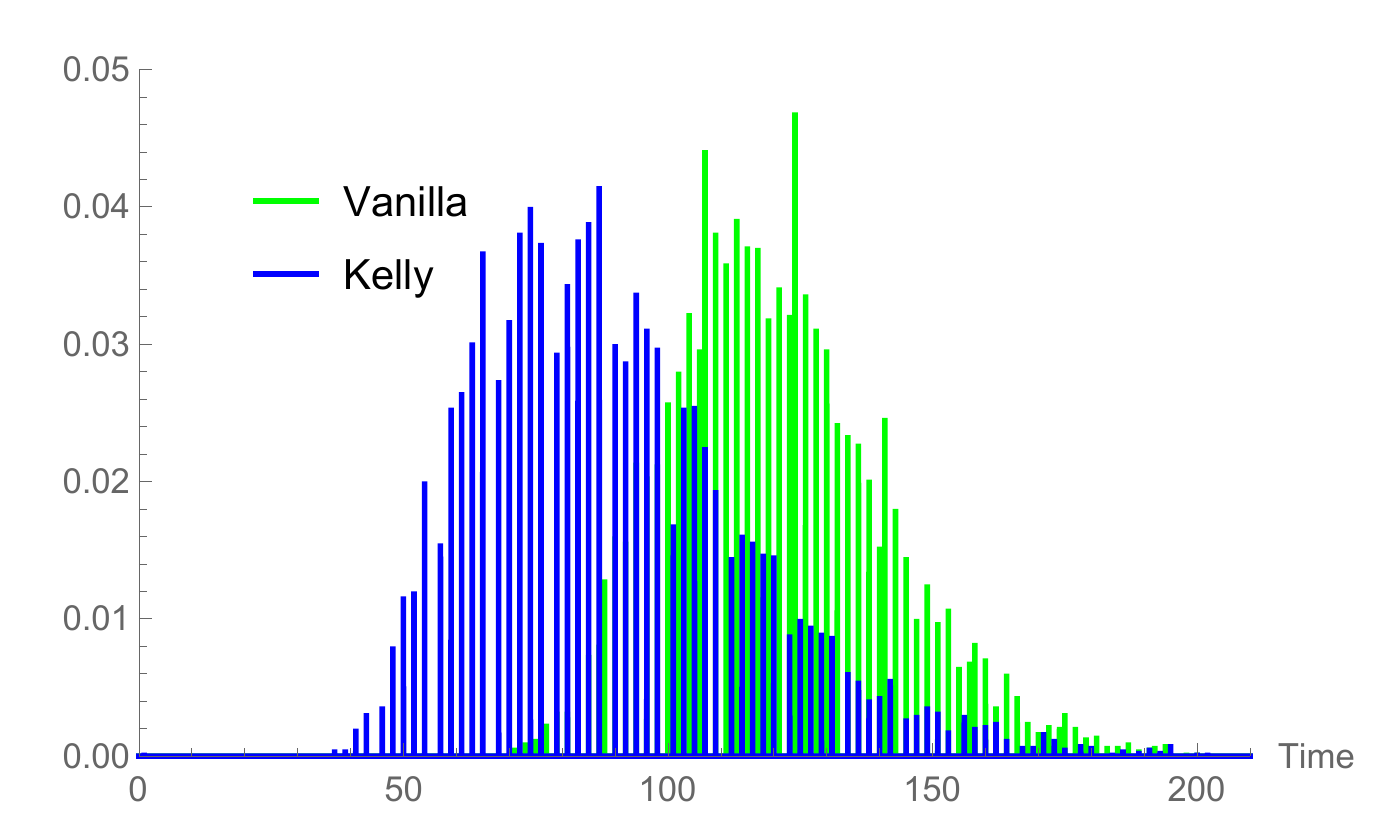}
\centering\includegraphics[width=7cm]{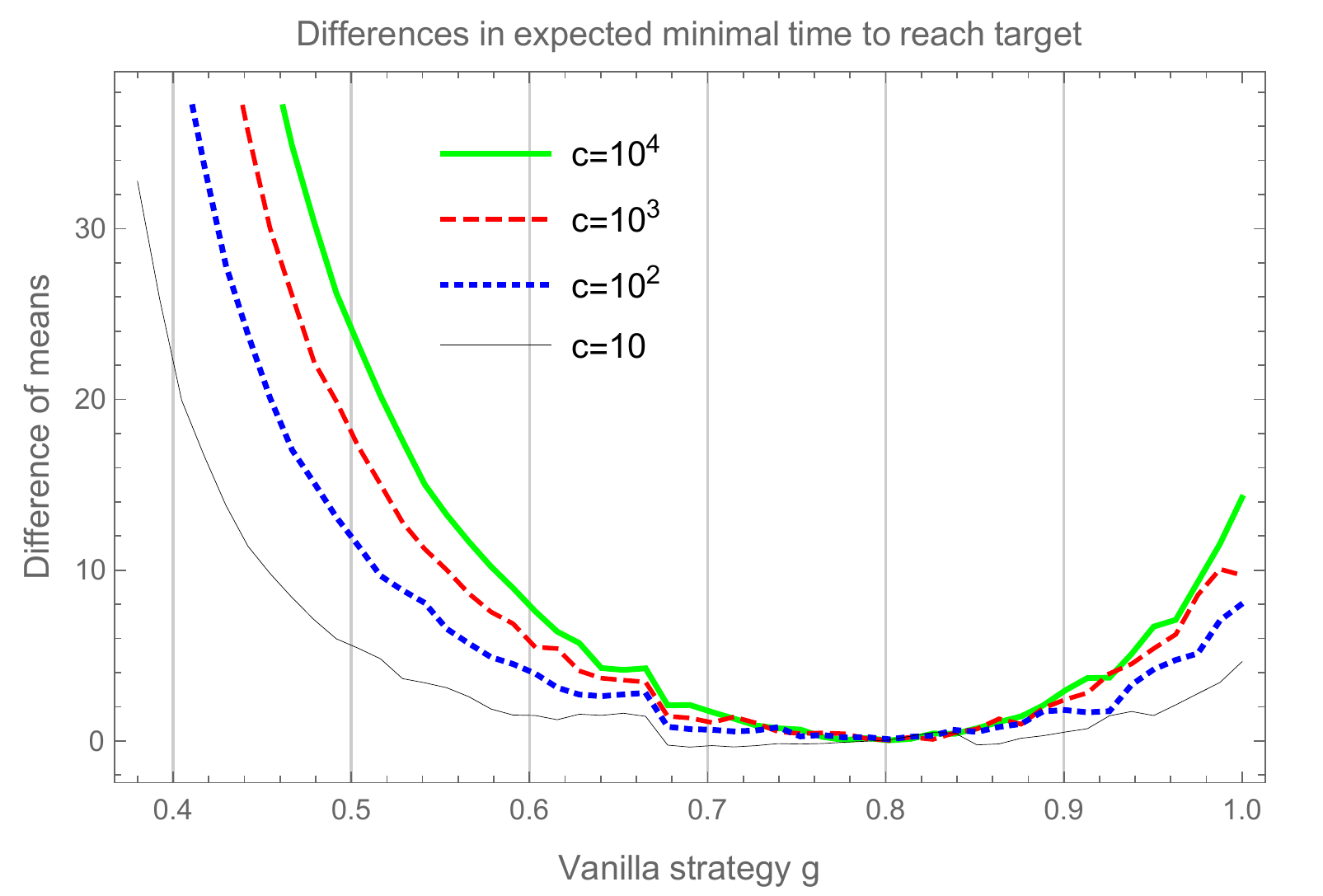}
\caption{{\small A: An example of the marginal distributions of the minimal time to reach a certain population size threshold, for two strategies | B: A simulation of $(E[T(g,c)]-E[T(f^*,c)])$ for several target values $c$, under the range of vanilla strategies $g$ corresponding to positive asymptotic growth rates, where $f^*=0.8$.}}
\label{fig:7}
\end{figure}\\  

Instead of considering expectations of (highly correlated) minimal time distributions, we devise a more informative fitness payoff function based on the joint distribution. Crucially, this payoff will naturally be amenable to a game-theoretic approach, in line with the type of analysis in the previous section with payoff $M_n (f,g)$. As before, we condition the probability on avoiding an extinction threshold. The payoff captures the probability that a trajectory following strategy $f$ reaches threshold $c$ before a trajectory following strategy $g$, conditioned on avoiding an extinction threshold $d$. If both trajectories reach $c$ at the same time (since time is in discrete generations) then the one which overshoots with a greater margin above $c$ `wins'. Denote by $T(f,c)$ the minimal time distribution given strategy $f$ and target lineage size $c$, 
\[
\begin{split}
M_c (f,g)=P(T(f,c)<T(g,c)  | \text{ extinction level } d).     
\end{split}	
\]
More precisely, we denote new trajectories $\{W^E_k\}_{k=1}^n$ by 
\[
W_0^E \in U[a,b],
\]
and for all $k=0,\ldots, n-1$
\[
\quad W^E_{k+1} =  \begin{cases}W^E_{k} \overline{o}_1(f)^{x_{k+1}} \overline{o}_2(f)^{1-x_{k+1}}, \text{ if } W^E_{k} \ge E\\
0, \text{ if } W^E_{k} < E.
\end{cases}
\]
We denote also by
\[
T(f,c) := \min\{n: W^E_n(f) \ge c\} 
\]
the first time when the trajectory $\{W^E_k\}_{k=0}^n$ cut the threshold $c$. $T(f,C) = \infty$ if and only if this trajectory does not cut the threshold.

Then the payoff matrix $M_c(f,g)$ is defined by
\bel{eq:6}
\begin{split}
M_c(f,g) &:= \PP(T(f,C) < T(g,C) | \text{given that at least one of them is finite})\\
&\quad + \PP(T(f,C) = T(g,C), W^E_{T(f,C)} > V^E_{T(g,C)} | \text{given that at least one of them is finite}).
\end{split}
\qe

We then identify pure strategy Nash equilibria reflecting the evolutionary solutions with the new relative payoff $M_c (f,g)$. In Appendix~\ref{app:L} we prove that again Kelly is the solution to the game, invariant to the evolutionary `choice' of target population size $c$, and that under a nonstationary environment regime Kelly emerges as the static equilibrium strategy. Finally, we highlight a deep mathematical link of this probabilistic perspective for minimal time optimality to the competitive optimality setting with payoff $M_n (f,g)$. Formally, $M_c(f,g)$ can be rewritten as a convex linear combination of $M_n(f,g)$: $M_c(f,g) = \sum_{n=0}^{\infty} P( W_0 W_n(f) > V_0 W_n(g),   T(f,c) = n)$ (see Appendix~\ref{app:L} for more details).

\section{Discussion}

\noindent
In this work we provide further support for the robustness of the expected log criterion as an optimality solution for biological bet hedging. We develop a game-theoretic framework inherently invariant to the span of evolutionary horizons while incorporating considerations of interim extinction risk, and use multiple optimality criteria to strengthen our results. This approach goes beyond standard models of bet-hedging, which focus on indefinite `long-term' growth rates and that ignore accounting for interim risk. Previous work generally upholds that ``phenotypes with the greatest long-term average growth rate will dominate the entire population'' as ``the basic principle'' used in optimization (\cite{yj96}), or that a proxy for the likely outcome of evolution is ``to think of organisms as maximizing the long-term growth rate of their lineage'' (\cite{donaldson10}). 

Nevertheless, some authors have recently acknowledged the importance of accounting for finite time horizons. For instance, \cite{rl11} note in passing that in their model ``the growth rate emerges as a unique measure of fitness when considering the long-term limit $T\to \infty$, but, if considering a finite ``horizon'', there may be a different strategy that outperforms [it]''. Indeed, as some evolutionists have argued, short-term fitness measures are also needed to achieve a full understanding of how evolution works in variable environments, as geometric mean fitness concerns the long-run evolutionary outcome (\cite{okasha18}).
Moreover, long-term fitness metrics are typically formulated without regard to transient short-term population dynamics, in which lineages might come close to extinction. Under more inclusive models with extinction, selection in a fluctuating environment can also favor bet-hedging strategies that ultimately increase the risk of extinction (\cite{lr19}). Given such considerations, the benefit of explicitly incorporating extinction considerations in stochastic growth models is clearly evident.

We have opted to focus on symmetric Nash equilibria rather than evolutionary stable strategies (ESS), which are strategies that cannot be beaten if the fraction of the rival invading mutants in the population is sufficiently small, and are generally invoked in settings with iterative match-ups between individuals rather than lineages (\cite{sp73}). Since the payoff in our game theoretic setting pits one lineage against another (two different strategies) there is no explicit sense of invading mutants  (but see \cite{olofsson09} for an ESS approach to bet-hedging). Moreover, some of the classic aspects of Nash's theorem do not directly apply within our setting. The theorem states that for every two-person zero-sum game with finitely many strategies there exists a mixed strategy that solves the game (\cite{nash51}). While our framework is indeed ``two-person'' it is not zero-sum and has finitely many strategies. Crucially, since an implicit goal of theoretical work such as ours may be towards predicting which strategies are likely to evolve, we focus on pure strategies rather than mixed ones, where the uniqueness of the equilibrium solution emerges as especially beneficial (echoing the classic approach of growth rate log-optimality where there is always a unique solution due to convexity).

We are not the first to attempt to model the expected minimal time to reach a finite asymptotic target, an extension of the seminal result of \cite{breiman61} on properties of the log optimal portfolio. \cite{aucamp77}) derived the first such analysis, given some basic assumptions that concern reaching a wealth target exactly vs. ``overshooting'' it. More recently, \cite{kp10} find that in a continuous time or asset price model where a finite target can be exactly reached with no overshooting, the Kelly solution is still optimal; in a discrete time model Kelly is only approximately optimal, but if ``time rebates'' are introduced (to compensate overshooting the goal in the last investment period) it becomes exactly optimal. While these results on the expectation of the time distribution are in line with our analysis of stochastic lineage growth optimality, we obtain an even stronger result: given {\it finite} population size targets, the log-optimal strategy emerges as a Nash equilibrium under a payoff function based on the {\it joint} distribution of minimal time trajectories.

Interestingly, \cite{Kelly56} has anticipated the application of his ideas in biological bet hedging, writing ``Although the model adopted here is drawn from the real-life situation of gambling it is possible that it could apply to certain other economic situations\ldots  the essential requirements for the validity of the theory are the possibility of reinvestment of profits and the ability to control or vary the amount of money invested or bet in different categories.'' It does not require a leap of the imagination to notice analogies of ``economic situations'' to evolutionary strategies, of ``reinvestment of profits'' to biological reproduction and growth, and of the ``control'' of invested money to evolved adaptative optimality. Of course, it best appreciated with Shannon's famous ``bandwagon'' warning in mind, cautioning over hasty attempts to apply insights from information theory to other fields (\cite{shannon56}).

\subsection{Other approaches to optimization under finite horizon and risk}

\noindent
A seemingly straightforward way of introducing finite (albeit still arbitrary) horizons into optimization settings is by considering the expectation of a finite-horizon growth rate. This is the approach adopted in some recent stock portfolio models for finite horizons (\cite{vz13}; \cite{morgan15}). Within our formalism from Eq.~\eqref{eq:2}, this amounts to finding,
\[
\argmax_f E\Big[W_n(f)^{\frac{1}{n}}\Big] = \argmax_f \Bigg(\sum_{i=1}^k \Big(\sum_{j=1}^k f_j o_{ij}  \Big)^{\frac{1}{n}} p_i \Bigg)^n.
\]
However, this implicitly assumes some arbitrary utility function, in this case the $n$-th root, the maximization of which requiring some justification. In contrast, Kelly's focus on $\argmax_f E[\log W_n]$ while implicitly assumes logarithmic utility, is equivalent the limit of the above expression, and leads to desired optimality properties as famously laid out by \cite{breiman61}.  

A more convincing approach to maximizing wealth with risk management over finite horizons was proposed in \cite{rujeerapaiboon15} for portfolio construction. The authors consider the optimization of a minimum bound for finite-horizon growth,
\[
 \argmax_f \Big\{\argmax_c \PP\Big(\frac{1}{n} \log W_n\ge c\Big)\ge 1-\eps\Big\}
 \]
 with a degree of freedom corresponding roughly to a risk-aversion or a choice of certainty parameter. 
 
The expression above allows deriving the portfolio giving the highest minimum bound for wealth for any level of certainty $\eps$. While choosing a particular horizon $n$ and a risk-aversion parameter is perfectly sensible in an investment setting, the translation to the biological framework is problematic: what would be evolution's risk aversion in this setting? Or the appropriate time horizon for optimization? Any choice of these two parameters would inescapably be arbitrary in nature. In an alternative approach \cite{rujeerapaiboon18} reformulate the Kelly gambling setting in terms of the Conservative Expected Value (CEV), a risk-averse expectation for highly skewed distributions. This amounts essentially to devising a systematic way of constructing fractional Kelly strategies such that it is strongly coupled with the infimum of the finite-horizon growth rate. Here again, there is an implicit arbitrariness in the choice of horizon length if applied in the context of an evolutionary framework, which we seek to avoid.

Other authors have focused on incorporating risk to the standard Kelly gambling setting with an infinite time horizon. For instance, \cite{busseti16} develop a systematic way to trade off growth rate and drawdown risk by formulating a risk-constrained Kelly gambling problem within the standard setting of growth rate maximization under asymptotic horizons. The additional risk constraint limits the probability of a drawdown to a specified level. Nevertheless, for our purposes, percentage drawdown is arguably not a natural metric for representing lineage extinction risks, as compared with explicit extinction thresholds, especially in scenarios of competing finite-size populations (\cite{ashby17}). Still other approaches may seek to target risk minimization as a primary criterion. In an approach akin to our Dutch book analysis, \cite{wolf05} minimize the growth rate variance and consequently the probability of extinction due to `unlucky' environmental trajectories. However, this is at the inevitable expense of achieving high stochastic growth rates, a vital aspect of evolutionary fitness.

\subsection{Game-theoretic competitive optimality of Bell and Cover}

\noindent
The results presented here can also be seen as both a special case and an extension of the classic results of \cite{bc80, bc88}. There are several important distinctions: [a] their setting is formulated for continuous random variables whereas our environments are discrete events, [b] their payoff implies a zero-sum game whereas our game is non zero-sum (more accurately, non-constant-sum) due to the effect of extinctions, and [c] their payoff function is a straightforward probability while our payoff is effectively a conditional probability (includes considerations of extinction risk). Moreover, implicit in Bell and Cover's setting is an infinitely sized payoff matrix, whereas our payoff matrix is finite since it reflects a finite number of strategies possible in a finite population. These distinctions have enabled us to show that, at least given the particular payoff function and discrete framework, the emerging symmetric Nash equilibrium is in fact a strict and unique one. 

Some authors have generalized or utilized other aspects of the classic competitive optimality results. Most recently, \cite{garivaltis18} has shown that discrete-time results of \cite{bc88} hold equally well for continuous-time rebalanced portfolios in a competitive setting between two investors, each aiming to maximize the expected ratio of one's own wealth to the other. In an original use of evolutionary ideas in finance, \cite{lo17} and \cite{orr17} consider a payoff function capturing relative wealth of two competing investors each with some set initial wealth, focusing on finite-period analysis. They analyze optimal strategies of a primary player against a given `vanilla' strategy, a framework consistent with our initial relative payoff non-game-theoretic setting. They find that the particular vanilla strategy chosen plays an important role in the optimal allocation, in conjunction with initial wealth of both players.

Finally, our game-theoretic analysis may hint at a solution to a ``coincidence'' pointed out in \cite{bc80}. They were left perplexed as to why competitive optimality for a finite horizon turned out, by ``coincidence'', to have the same solution (namely, Kelly) as in the growth-optimal portfolio: ``Finally, it is tantalizing that $b^*$ arises as the solution to such dissimilar problems [\ldots] The underlying for this coincidence will be investigated''. Their follow-up 1988 paper suggests a ``possible reason for the robustness of log optimal portfolios'' or why ``log optimal portfolios behave well in the competitive investment game'': namely that the wealth generated from any portfolio is always within ``fair reach'' of the wealth from the log-optimal portfolio. Indeed, the Kuhn-Tucker conditions and the consequent bound on the wealth ratio (\cite[Theorem 16.2.2]{ct06}) already imply that game-theoretic optimality is the driving force behind the asymptotic dominance. Fair randomization of initial wealth then leads to the game-theoretic solution for any increasing function of the wealth ratio. Our investigation of the payoff matrix suggests another perspective to this ``coincidence''. Asymptotically with horizon $n$, the payoff matrix becomes maximally `contrasted', with off-diagonal cells converging to probabilities of $0$ or $1$ (except those on `fault lines'), such that the Nash equilibrium emerges naturally. In effect, the `saddle-point' equilibrium, which has been established as invariant with $n$, asymptotically attains maximum curvature (Appendix~\ref{app:M}).

\section{Conclusion}

\noindent
In this work we have argued that under fluctuating environments and trait randomization geometric mean fitness should also encompass considerations of stochastic growth and extinction risk under finite evolutionary horizons. We show that for both the relative maximal growth payoff and the relative minimal time payoff there is a unique pure-strategy symmetric equilibrium, which is invariant with evolutionary time horizon and robust to low extinction risk. Coinciding with the classic bet-hedging modeling approach, this is the Kelly log-optimal strategy. With higher thresholds of extinction, the equilibrium may shift away from Kelly and possibly branch out to multiple equilibria. Future work will be required to generalize the model to competitive optimality payoffs beyond pairwise lineages, Markovian environmental sequential transitions, random fitness matrices, and to more precisely capture the effect of high extinction thresholds on the optimal evolutionary solutions.

\section*{Acknowledgements}

\noindent
We'd like to thank Alex Garivaltis for illuminating discussions on competitive optimality and two diligent reviewers for their very insightful comments. We also appreciate the continued support of J\"urgen Jost and the Max Planck Institute (MIS). OT would like to further acknowledge the generous support of the Complexity Institute at NTU Singapore and Peter MA Sloot. TDT would also like to thank VIASM for financial support and hospitality in his two-month visiting in 2019.

\newpage
\begin{appendices}
\counterwithin{figure}{section}

\small

\section{The Kelly solution to the full fitness matrix model}\label{app:A}
In this section, we derive the Kelly (log optimal) solution for the full-fitness matrix model. 

\noindent
{\bf The case $k=2$:} We have
\[W_n(f) = (o_{11}f+o_{12}(1-f))^H (o_{21}f+o_{22}(1-f))^{n-H}\]
where $H\sim Binomial(n,p)$. The Kelly solution is then defined by  
\[
f^{Kelly} := \argmax_{f\in [0,1]} G(f) 
\]
where $G(f) := \lim_{n\to \infty}  W_n^{\frac{1}{n}}(f)$. \\
By denoting $\overline{o}_1(f):=o_{11}f+o_{12}(1-f), \overline{o}_2(f):=o_{21}f+o_{22}(1-f)$, we have
\[
\begin{split}
\lim_{n\to \infty}\frac{1}{n} \log W_n(f) &= \lim_{n\to \infty} \Bigg(\frac{H}{n} \log \overline{o}_1(f) + (1-\frac{H}{n}) \log \overline{o}_2(f)\Bigg)= p \log \overline{o}_1(f) + (1-p) \log \overline{o}_2(f).
\end{split}
\]
Therefore, by directed calculations, we obtain the Kelly solution which is dependent on $p$
\[
f^{Kelly}(p) = \begin{cases}
0,& \text{ if } p\in [0,p_{-}]\\ 
\frac{(1-p) o_{12}}{o_{12}-o_{11}} + \frac{p o_{22}}{o_{22}-o_{21}},& \text{ if } p\in [p_{-}, p_{+}]\\
1,& \text{ if } p\in [p_{+},1],
\end{cases}
\]
where $p_{-} = \frac{o_{12}(o_{22}-o_{21})}{\De}$ and $p_{+}= \frac{o_{11}(o_{22}-o_{21})}{\De}$,  
and the corresponding optimal value is
\[
G(f^{Kelly})
= \begin{cases}
o_{12}^p  o_{22}^{1-p},& \text{ if } p\in [0,p_{-}]\\ 
\Big(\frac{p(o_{11}o_{22}-o_{12}o_{21})}{o_{22}-o_{21}}\Big)^p \Big(\frac{(1-p)(o_{11}o_{22}-o_{12}o_{21})}{o_{11}-o_{12}}\Big)^{1-p},& \text{ if } p\in [p_{-}, p_{+}]\\
o_{11}^p  o_{21}^{1-p},& \text{ if } p\in [p_{+},1].
\end{cases}
\]

\noindent
{\bf The case general $k$:}
By directed calculations, we obtain
\[
G(\f) =  \prod_{i=1}^k \overline{o}_i(\f)^{p_i}
 \] 
where $\overline{o}_i(\f) := \sum_{j=1}^k o_{ij} f_j$. This implies that for each $\p \in \De_{k-1}:=\{(x_1,\ldots,x_k)  \in [0,1]^k \text{ such that } x_1+\cdots+x_k =1\}$, $G(\f)$ is a continuous strict convex function in the compact convex domain $\De_{k-1}$. Therefore there will always exist a unique Kelly solution $f^{Kelly} \in \De_{k-1}$ which is dependent on $\p$. 

\noindent
{\it Remark:}
\begin{enumerate}
\item[(i)] If the fitness matrix is diagonal, i.e., $(o_{ij}) = \diag\{o_1,\ldots,o_k\}$, then $(f^{Kelly})_i = p_i$;  
\item[(ii)] $f^{Kelly}$ solves the system 
\[
\sum_{i=1}^k \frac{p_i o_{ij}}{\overline{o}_i(\f)} = 1, \quad \forall   j=1,\ldots, k.
\]
\end{enumerate}


\section{The solution to nonstationary environments}\label{app:B}
We model the environment probabilities on a parameterized Beta distribution, such that $p \sim B(\al,\beta)$, and prove that the Kelly solution (a static $f$ that maximizes the asymptotic growth rate) in the asymptotic framework corresponds to the solution of the i.i.d. environment case with a probability equaling the expectation of the Beta distribution. \\
For sake of simplicity, we consider only $k=2$. We have
$
W_n(f) = \overline{o}_1(f)^H \overline{o}_2(f)^{n-H} 
$,
where $H\sim GB\big(n,\{p_1,\ldots, p_n\}\sim Beta(\al,\beta)\big)$, i.e., $H=\eps_1+\cdots+\eps_n$ with $\eps_r \sim Bernoulli(p_r)$ and $p_r \sim Beta(\al, \beta)$. Using the law of large numbers, we have
\[
\begin{split}
G(f) &= \lim_{n\to \infty} \Bigg(\frac{\sum_{i=1}^n \eps_i}{n} \log \overline{o}_1(f)  + \Big(1- \frac{\sum_{i=1}^n \eps_i}{n}\Big)  \log \overline{o}_2(f)\Bigg)=\lim_{n\to \infty} \Bigg(\frac{\sum_{i=1}^n \E \eps_i}{n} \log \overline{o}_1(f)  + \Big(1- \frac{\sum_{i=1}^n \E\eps_i}{n}\Big)  \log \overline{o}_2(f)\Bigg)\\
&=\lim_{n\to \infty}\Bigg(\Big(\frac{1}{n}\sum_{r=1}^n p_r\Big) \log \overline{o}_1(f) + \Big(1-\frac{1}{n}\sum_{r=1}^n p_r\Big)\log \overline{o}_2(f)\Bigg)= p \log \overline{o}_1(f)+ (1-p)\log \overline{o}_2(f) \quad \text{ a.s.}
\end{split}
\]
where $p=\lim_{n\to \infty}\frac{1}{n}\sum_{r=1}^n p_r= \frac{\al}{\al+\beta}$ is the expectation of the Beta distribution. Thus, the Kelly solution in this case is the same as the previous case.
 
\section{The Dutch book solution and the corresponding loss of growth}\label{app:C}
In this section we derive the Dutch book solution for our model. By definition, the Dutch book solution $f^D$ satisfies $\overline{o}_1(f) = \overline{o}_2(f) = \cdots= \overline{o}_k(f)$ with the positive growth, i.e., $\overline{o}_1(f)>1$. 

\vspace{.2cm}
\noindent
{\bf The case $k=2$:}
The Dutch book solution satisfies \[
o_{11}f + o_{12}(1-f) = o_{21}f + o_{22}(1-f) > 1.
\]
Therefore, if $\De := o_{11}o_{22} - o_{12}o_{21} > o_{11}+o_{22}-o_{12}-o_{21}$ then we always have a unique Dutch book solution $f^D$ 
\[
f^D= \frac{o_{22}-o_{12}}{o_{22}-o_{12}+o_{11}-o_{21}}
\]
and
\[
G(f^D) = \frac{\De}{o_{22}-o_{12}+o_{11}-o_{21}} > 1 \quad \text{ (which does not depend on $p$).}
\]


\vspace{.2cm}
\noindent
{\bf The general case $k$:}
We give out here some criteria to have a unique Dutch book solution in the general case $k$.
\begin{lem}
Given a fitness matrix $O = (o_{i,j})_{i,j=1}^k$. Denote by $\al_{i,j} = o_{i,j}-o_{k,j}$ for all $j=1,\ldots,k$ and $i=1,\ldots,k-1$. Denote by $\Lambda = (\Lambda_{i,j})_{i,j=1}^k$ such that
\[
\begin{pmatrix}
\al_{1,1} & \cdots & \al_{1,k-1}& \al_{1,k} \\
\vdots & \ddots & \vdots & \vdots\\
\al_{k-1,1} & \cdots & \al_{k-1,k-1}& \al_{k-1,k} \\
1 & \cdots & 1 & 1 
\end{pmatrix}
\begin{pmatrix}
\Lambda_{1,1} & \cdots & \Lambda_{1,k} \\
\vdots & \ddots & \vdots\\
\Lambda_{k,1} & \cdots &  \Lambda_{k,k} 
\end{pmatrix}=
I_k.
\]
If this fitness matrix $O$ satisfies 
\begin{enumerate}
\item[(i)] $o_{ii} > o_{ji}\ge 0$ for all $i,j = 1,\ldots, k$
\item[(ii)] $\Lambda_{i,k} > 0$ for all $i=1,\ldots,k$
\item[(iii)] $\suml_{j=1}^k o_{i,j}\Lambda_{j,k} > 1$ for all $i=1,\ldots,k-1$
\end{enumerate}
then there exists a Dutch book solution defined by 
$
f^D_j =\Lambda_{j,k}, \quad j=1,\ldots,k
$
and the corresponding deterministic wealth is
\[
G(\f^D) = \prod_{i=1}^k \Big(\suml_{j=1}^k o_{i,j}f^D_j\Big)^{p_i} = \suml_{j=1}^k o_{i,j}\Lambda_{j,k} .
\]
\end{lem}
\begin{proof}
We have from Condition $(iii)$
\[
\overline{o}_1(\f^D) = \suml_{j=1}^k o_{1,j} f^D_j = \suml_{j=1}^k o_{1,j}\Lambda_{j,k} > 1;
\]
Moreover from the definition of $\al$ and $\Lambda$ we have $\overline{o_i}(\f^D) = \overline{o}_j(\f^D)$ for all $i\ne j=1\ldots,k.$

\end{proof}
\begin{cor}\label{cor:1}
In the case of a diagonal matrix, i.e., $o_{i,j} = \diag\{o_1,\ldots,o_k\}$, by direct calculation, we obtain 
$
\Lambda_{i,k} = \frac{o_i^{-1}}{\suml_{j=1}^k o_j^{-1}}.$\\
Conditions (i) and (ii) hold true iff $o_i>0$ and condition (iii) holds true iff $\suml_{j=1}^k o_j^{-1} < 1$.
\end{cor}
\begin{cor} For a finite $n$ and assuming
$
\min\{o_{ii}\}_{i=1}^k \gg \max\{o_{ij}\}_{i\ne j} \ge 0,
$
there exists a Dutch book solution $\f^D$.
\end{cor}

\begin{proof}
The conclusion directly follows from the above Corollary~\ref{cor:1} (for a diagonal fitness matrix). 
\end{proof}

\section{Finite last intersection} \label{app:D}
In this section, we show that for a given pair of strategies $(f, g)$ with $G(f) > G(g)$, there is a $T(f,g) < \infty$ such that $W_n(f,\x) > W_n(g,\x)$ for all $n \ge T(f,g)$ and for all $\x \in \{0,1\}^{\infty}$. This means that the last intersection between two random trajectories $\{W_n(f,\x)\}_n$ and $\{W_n(g,\x)\}_n$
\[
\tau(\x) := \sup\{n: W_n(f,\x) \le W_n(g,\x)\}
\]
 is bounded above by $T(f,g)$ (a finite number depending only on $f$ and $g$). 
\begin{proof} 
We first define the excess growth rate 
\[
E_n(\x) := \frac{1}{n} \log W_n(f,\x) - \frac{1}{n} \log W_n(g,\x).
\]
We note that for all $\x$
\bel{eq:cond}
\lim_{n\to \infty} E_n(\x) = \log G(f) - \log G(g) >0.
\qe 
To this end we need to prove that there is a $T(f,g) < \infty$ such that 
\[
\inf_{\x} E_n(\x) > 0 \quad \forall n\ge T(f,g). 
\]
Otherwise, for each $k$ there exist $n_k \ge k$ and $\x_k \in \{0,1\}^{\infty}$ such that $E_{n_{k}} (\x_k) \le 0$. Now, there exists a subsequence of $\{\x_k\}$ which is convergent to some $\x\in \{0,1\}^{\infty}$. Therefore as $k\to \infty$ we have $n_k\to \infty$ and $\lim_{n_k\to \infty}E_{n_{k}} (\x) \le 0$, in contradiction to \eqref{eq:cond}.
\end{proof}

\section{Asymptotic log-normality of the growth rate} \label{app:E}
In this section, we show that in our discrete model, the growth rate approaches log-normality with zero variance. 
\begin{proof}
We rewrite 
$	\frac{1}{n} \log W_n(f) = \frac{1}{n} \sum_{i=1}^n y_i,$ where $y_i = x_i \log  \overline{o}_1(f)  + (1-x_i) \log  \overline{o}_2(f) $ are independent discrete random variables with values: $\log  \overline{o}_1(f) , \log  \overline{o}_2(f) $ and probabilities: $p, 1-p$ correspondingly. Thus we have a sequence of i.i.d. random variables $\{y_i\}_i$ with expectation $\mu = E(y_i) = G(f)$ and variance $\sigma^2 = var(y_i) = p(1-p) (\log \overline{o}_1(f) - \log \overline{o}_2(f))^2$. By using the CLT, we have for a large $n$:
$
1/ \sqrt{n} \sum_{i=1}^n ( y_i - \mu) \sim N(0,\sigma^2)
$
which is equivalent to 
\[
\frac{1}{n} \sum_{i=1}^n y_i   \sim  N(\mu, \frac{\sigma^2}{n}).
\]  

\end{proof}

\section{Fully correlated log growth rates for the case k=2:} \label{app:F}
In this section we show that for all $f, g\ne f^D$
\[
\corr[\log W_n(f),\log  W_n(g)] = \pm 1.
\]

\begin{proof}
Denote by
\[
W_n(f, \x) = \overline{o}_1(f)^{|\x|}   \overline{o}_2(f)^{n-|\x|},
\]
where $\x = (x_1,\ldots,x_n)$ is a realization and $|x| = x_1+\cdots+x_n$. Because $f,g \ne f^D$ we have $\overline{o}_1(f)\ne \overline{o}_2(f)$ and $\overline{o}_1(g)\ne \overline{o}_2(g)$, therefore we can define
\[
\ld = \frac{\log \frac{\overline{o}_1(f)}{\overline{o}_2(f)}}{\log \frac{\overline{o}_1(g)}{\overline{o}_2(g)}} \in \R\setminus\{0\}.
\]
We first prove that for any given $m$ realizations $\x^{(1)}, \ldots, \x^{(m)}$, we have
\bel{eq:ld}
\log W_n(f,\x^{(i)}) - \frac{1}{m} \sum_{k=1}^m \log W_n(f,\x^{(k)}) = \ld \Bigg(\log W_n(g,\x^{(i)}) - \frac{1}{m} \sum_{k=1}^m \log W_n(g,\x^{(k)})\Bigg).
\qe
Indeed, we note that
\[
\begin{split}
\log W_n(f,\x^{(i)}) - \log W_n(f,\x^{(j)})  = \log \frac{\overline{o}_1(f)^{|\x^{(i)}|}   \overline{o}_2(f)^{n-|\x^{(i)}|}}{ \overline{o}_1(f)^{|\x^{(j)}|}   \overline{o}_2(f)^{n-|\x^{(j)}|}}
= (|\x^{(i)}| - |\x^{(j)}|)\log \frac{\overline{o}_1(f)}{\overline{o}_2(f)},
\end{split}
\]
and similarly for $g$.
This implies \eqref{eq:ld}. Therefore
\[
\begin{split}
\corr&[\log W_n(f),\log  W_n(g)] = \frac{\cov[\log W_n(f),\log  W_n(g)]}{\sqrt{\var[\log W_n(f)]}\sqrt{\var[\log W_n(g)]}}\\
&= \frac{\frac{1}{m} \sum\limits_{i=1}^m (\log W_n(f,\x^{(i)}) - E[\log W_n(f)])(\log W_n(g,\x^{(i)}) - E[\log W_n(g)])}{\sqrt{\frac{1}{m} \sum\limits_{i=1}^m \Bigg(\log W_n(f,\x^{(i)}) - E[\log W_n(f)]\Bigg)^2}\sqrt{\frac{1}{m} \sum\limits_{i=1}^m \Bigg(\log W_n(g,\x^{(i)}) -E[\log W_n(g)])\Bigg)^2}}=\frac{\ld}{|\ld|} = \pm 1.
\end{split}
\]
\noindent
\end{proof}

\begin{rem}
Whether the correlation is $\pm1$ depends on $\lambda > 0$ or $\lambda < 0$. For $f=f^{Kelly}$ the growth factor with  environment ``1'' $>$ the growth factor with environment ``0'' implying $\log \frac{ \overline{o}_1(f) } { \overline{o}_2(f)} > 0$. Similarly for $g$ it implies $\log \frac{ \overline{o}_1(g) } { \overline{o}_2(g)} > 0$, therefore $\lambda > 0$. At $f^D$, $\log \frac{ \overline{o}_1(f^D) } { \overline{o}_2(f^D)} = 0$, therefore it acts as a threshold. In most cases the correlation will be $+1$ since both $f$ and $g$ induce a positive growth rate. 
\end{rem}

\section{Kelly is the maximal element in the fitness payoff relation} \label{app:H}
Here we assume lineage size initial randomization, i.e., $W_n(f) \gg W_n(g)$ iff 
\[
M_n(f,g):= \PP(W_0 W_n(f) > V_0 W_n(g)) \ge \PP(V_0 W_n(g)>W_0 W_n(f))
\]
where $W_0$ and $V_0$ are random, and show that the Kelly strategy is the maximal element in this relation. 
\label{prop:sumoneinv}
\begin{proof}
As a direct consequence of Proposition~\ref{prop:sumone} and Eq.~\eqref{eq:max} we have
\[
M_n(f^{Kelly},f) \ge  1/2 \ge M_n(f,f^{Kelly}) \quad \forall f \in [0,1],
\]
and equality if and only if $f=f^{Kelly}$. 
\end{proof}

\section{Non-constant-sum game, but conceptually zero-sum} \label{app:J}

In this section we show that 
\begin{prop}\label{prop:sumone}
\begin{enumerate}
\item[(i)] For $d=0$, $M_n(f,g)+M_n(g,f) = 1$ for all $f,g$.
\item[(ii)] For $d> 0$, $M_n(f,g)+M_n(g,f) < 1$ for all $f,g$.
\end{enumerate}
\noindent
Moreover, the game is conceptually zero-sum, but not formally. 
\end{prop}
\begin{proof}
(i) We have from Eq.~\eqref{eq:Mn}
\[
\begin{split}
M_n(f,g)+M_n(g,f) &=  \sum_{s=0}^n  \Big(\PP(W_n(f,s) W_0 > W_n(g,s) V_0)+ \PP(W_n(f,s) W_0 < W_n(g,s) V_0) \Big)P(s) =\sum_{s=0}^n P(s)  = 1.
\end{split}
\]
(ii) On the other hand, we have from Eq.~\eqref{eq:5} for all $f\ne g$
\[
\begin{split}
M_n(f,g)+M_n(g,f) &= \PP(CAB) + \PP(AB^c) + \PP(C^c AB) + \PP(BA^c)= \PP(AB) + \PP(AB^c) + \PP(BA^c) = \PP(A\cup B) < 1.
\end{split}
\]
where $C = \{W_0 W_n(f) > V_0 W_n(g)\}$, $A= \{W_0 W_i(f) > d \quad \forall i=1,\ldots,n\}$, $B= \{V_0 W_i(g) > d \quad \forall i=1,\ldots,n\}$.\\
For $f=g$ we also have
\[
\begin{split}
M_n(f,f)&= \PP(W_0 > V_0, A_1, A_2)< \PP(W_0 > V_0) =\frac{1}{2}.
\end{split}
\]
where $A_1= \{W_0 W_i(f) > d \quad \forall i=1,\ldots,n\}$, $A_2= \{V_0 W_i(f) > d \quad \forall i=1,\ldots,n\}$.\\
Finally, numeric simulations demonstrate that if $M(W,V)>M(U,V)$ then $M(V,W)<M(V,U)$ for all $W,V,U$, i.e. changing to a strategy with a gain for one player always incurs a loss for the other player.
\end{proof}

\section{The symmetric Nash equilibrium solution to payoff $M_n(f,g)$} \label{app:K}
\begin{prop}\label{prop:cond}
We always have
\[
\E\Bigg(\frac{W_n(f)}{W_n(f^{Kelly})}\Bigg) \le 1
\]
and the equality happens if and only if $p_{-}<p<p_{+}$.
\end{prop}
\begin{proof}
For given $f,g$, we denote by $\al_1 = \frac{\overline{o}_1(f)}{\overline{o}_1(g)}$, $\al_2 = \frac{\overline{o}_2(f)}{\overline{o}_2(g)}$. We have 
\[
\begin{split}
\E\Bigg(\frac{W_n(f)}{W_n(g)}\Bigg) &= \sum_{\x} \frac{W_n(f,\x)}{W_n(g,\x)} P(\x)= \sum_{s=0}^n \frac{W_n(f,s)}{W_n(g,s)} P(s)= \sum_{s=0}^n \al_1^s \al_2^{n-s} \binom{n}{s} p^s (1-p)^{n-s}= (p \al_1 +(1-p)\al_2)^n.
\end{split}
\]
On the other hand, from the formula 
$
f^{Kelly} = \begin{cases}
0,& \text{ if } p\in [0,p_{-}]\\ 
\frac{(1-p) o_{12}}{o_{12}-o_{11}} + \frac{p o_{22}}{o_{22}-o_{21}},& \text{ if } p\in [p_{-}, p_{+}]\\
1,& \text{ if } p\in [p_{+},1],
\end{cases}
$ we have for any pair $(f,f^{Kelly})$, $p \al_1 +(1-p)\al_2 = 1$ if $p\in [p_{-}, p_{+}]$ and $p \al_1 +(1-p)\al_2 < 1$ if $p\notin [p_{-}, p_{+}]$. 
\end{proof}

\begin{prop}\label{prop:3}
We consider a game with payoff without extinction
\[
M_n(f,g) := \PP(W_n(f) W_0 > W_n(g) V_0), 
\]
where $W_0, V_0$ have the same distribution.
Then, in this game, $(f^{Kelly},f^{Kelly})$ is a strict Nash equilibrium.
\end{prop}
\begin{proof}
First, we note that
\bel{eq:Mn}
\begin{split}
M_n(f,g) &=\sum_{s=0}^n  \mathbf \PP(W_n(f,s) W_0 > W_n(g,s) V_0) P(s)=   \sum_{s=0}^n  \mathbf \PP(\al_1^s \al_2^{n-s} W_0 > V_0) P(s)\\
&= \sum_{s \in A_1} P(s) \frac{1}{2} \al_1^s \al_2^{n-s} + \sum_{s \in A_2} P(s) \Bigg(1- \frac{1}{2} \al_1^{-s} \al_2^{-n+s}\Bigg),
\end{split}
\qe
where $A_1 = \{s\in \{0,\ldots,n\}: \al_1^s \al_2^{n-s}<1\}$ and $A_2 = \{0,\ldots,n\} - A_1$.
Therefore, for $f=g$ we have $\al_1=\al_2=1$, which implies $A_1=\emptyset$, $A_2=\{0,\ldots,n\}$ and
\bel{eq:equalhalf}
\begin{split}
M_n(f,f) =\sum_{s=0}^n P(s) \Bigg(1- \frac{1}{2}\Bigg)=1/2.
\end{split}
\qe
For any $f\ne f^{Kelly}$, by using the Cauchy inequality for the second term, we have

\[
\begin{split}
M_n(f,f^{Kelly}) &<  \sum_{s \in A_1} P(s) \frac{1}{2} \al_1^s \al_2^{n-s} + \sum_{s \in A_2} P(s) \frac{1}{2} \al_1^s \al_2^{n-s}= \frac{1}{2} (p\al_1+ (1-p)\al_2)^{n}.
\end{split}
\]
From Proposition~\ref{prop:cond} we have 
\bel{eq:max}
M_n(f,f^{Kelly}) < 1/2 = M_n(f^{Kelly},f^{Kelly}), \quad \forall f\ne f^{Kelly}.
\qe
Therefore $(f^{Kelly},f^{Kelly})$ is a strict Nash equilibrium.

\end{proof}

\begin{prop}\label{prop:4}
The above Nash equilibrium is the unique one in the game.
\end{prop}

\begin{proof}
Assume that $(f_0,g_0) \ne (f^{Kelly}, f^{Kelly})$ is another Nash equilibrium. Without loss of generality we assume that $g_0 \ne f^{Kelly}$. By definition of a Nash equilibrium, we have $M_n(f_0,g_0) \ge M_n(f,g_0)$ for all $f$ and $M_n(g_0,f_0) \ge M_n(g,f_0)$ for all $g$. By choosing $f=g=f^{Kelly}$ and using Propposition~\ref{prop:sumoneinv} we have
$M_n(f_0,g_0) \ge M_n(f^{Kelly},g_0) > \frac{1}{2}$ and $M_n(g_0,f_0) \ge M_n(f^{Kelly},f_0) \ge \frac{1}{2}$. This implies that $M_n(f_0,g_0) + M_n(g_0,f_0) > 1$ which is a contradiction to Proposition~\ref{prop:sumone}. Therefore $(f^{Kelly}, f^{Kelly})$  is the unique Nash equilibrium (see Fig.~\ref{fig:AK} where the equilibrium lies at the saddle-point of the payoff landscape.)
\end{proof}
\noindent

\begin{figure}[h]
\centering\includegraphics[width=6cm]{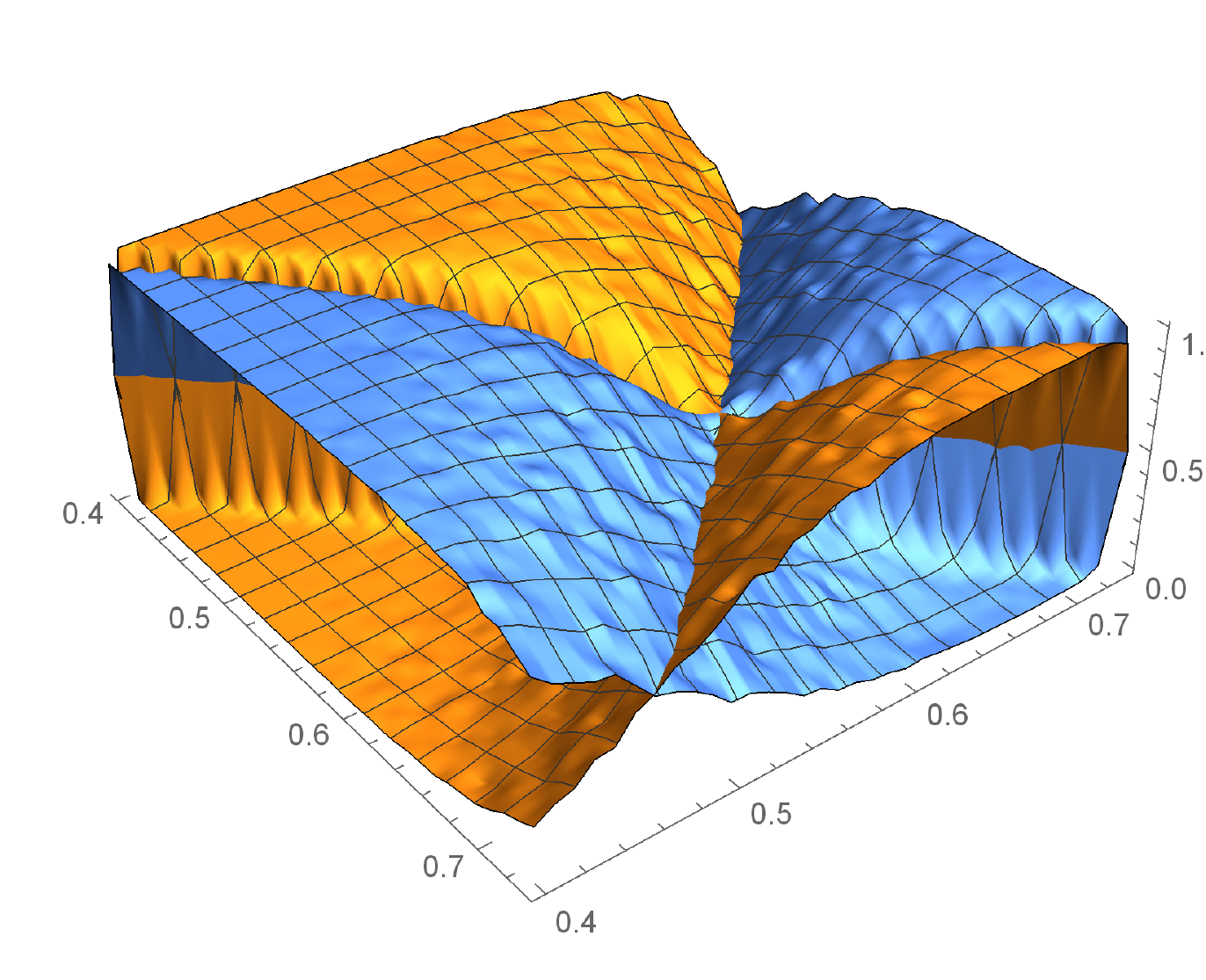}
\caption{{\small A 3D graphical representation of a probability payoff matrix (Eq.\eqref{eq:5})) with primary (blue) and opponent (orange) lineage
payoffs as intersecting saddle-point surfaces, highlighting the equilibrium locus.
}}
\label{fig:AK}
\end{figure}

\section{The symmetric Nash equilibrium solution to payoff $M_c (f,g)$} \label{app:L}
\begin{prop}
We consider a game with payoff defined as \eqref{eq:6} without extinction
\[
M_c(f,g) := \PP(T(f,c) < T(g,c)) + \PP(T(f,c) = T(g,c), W_0W_{T(f,c)}(f)> V_0W_{T(g,c)}(g)). 
\]
Then, in this game, $(f^{Kelly},f^{Kelly})$ is a strict Nash equilibrium.
\end{prop}
\begin{proof}
First we note that
\[
\begin{split}
M_c(f,g) &= \sum_{n=1}^{\infty}  \PP(T(g,c)>n, T(f,c) = n) + \PP(T(g,c)=n, W_0W_n(f) > V_0W_n(g), T(f,c) = n)\\
&= \sum_{n=1}^{\infty}  \PP(V_0W_n(g) < c, T(f,c) = n) + \PP(V_0W_n(g) \ge c, W_0W_n(f) > V_0W_n(g), T(f,c) = n)\\
&= \sum_{n=1}^{\infty}  \PP(W_0W_n(f) > V_0W_n(g), T(f,c) = n).\\
\end{split}
\]
Then, from Propposition~\ref{prop:3} we have 
\[
M_c(f,f) = \sum_{n=1}^{\infty}  \PP(W_0 > V_0, T(f,c) = n) = \PP(W_0 > V_0) = \frac{1}{2} \quad \forall f
\]
and 

\[
M_c(f,f^{Kelly}) < \frac{1}{2} \sum_{n=1}^{\infty}  \PP(T(f,c) = n) = \frac{1}{2} \quad \forall f.
\]
Therefore $(f^{Kelly},f^{Kelly})$ is a strict Nash equilibrium. 
\end{proof}
\begin{prop}
$(f^{Kelly},f^{Kelly})$ is the unique Nash equilibrium.
\end{prop}

\begin{proof}
We first note that for all $f,g$
\[
\begin{split}
M_c(f,g) +  M_c(g,f)&=\sum_{n=1}^{\infty}  \Big(\PP(W_0W_n(f) > V_0W_n(g) + \PP(W_0W_n(f) < V_0W_n(g)\Big), T(f,c) = n)\\
&=  \sum_{n=1}^{\infty}  \PP(T(f,c) = n) = 1.
\end{split}
\]
The left hand side is similar to the proof in Propposition~\ref{prop:4}. 
\end{proof}
\noindent
It is worthwhile here to highlight a link between the this payoff and $M_n(f,g)$. Formally, $M_c(f,g)$ can be rewritten as a convex linear combination of $M_n(f,g)$: 
\[
M_c(f,g) = \sum_{n=0}^{\infty} P( W_0 W_n(f) > V_0 W_n(g),   T(f,c) = n ).
\] 
This has a straightforward interpretation: for each event $(T(f,c) = n)$, [a] the event $(T(f,c) < T(g,c))$ is equivalent to the event $(T(g,c) > n)$ or $(V_0 W_n(g) < c <= W_0 W_n(f))$, and [b] the event $(T(f,c) = T(g,c), W_0 W_{T(f,c)} > V_0 W_{T(g,c)}(g))$ is equivalent to the event $(c <= V_0 W_n(g) < W_0 W_n(f))$. Consequently the combination of the two events $(T(f,c) < T(g,c))$ and $(T(f,c) = T(g,c), W_0 W_{T(f,c)} > V_0 W_{T(g,c)}(g))$ is equivalent to the event $(W_0 W_n(f) > V_0 W_n(g))$.

\section{The probability payoff matrix converges with horizon $n$ to the expected log matrix}\label{app:M}
\begin{prop}
For any pair $(f,g)$ with $G(f)\ne G(g)$, we have 
\[
M_{\infty}(f,g) := \lim_{n\to \infty} M_n(f,g) = \begin{cases}
1, \text{ if } G(f) >G(g)\\
0, \text{ if } G(f) < G(g).
\end{cases}
\]
\end{prop}
\begin{proof}
If $G(f)-G(g) = \eps>0$, then by a similar argumentation as Appendix~\ref{app:D}, there exists $n_0 < \infty$ such that for all $n\ge n_0$ and all $\x$
\[
0< \frac{1}{n} \log W_0 < \frac{\eps}{4}, \quad  0< \frac{1}{n} \log V_0 < \frac{\eps}{4}, \quad 
 \frac{1}{n} \log W_n(f,\x) > G(f) - \frac{\eps}{4}, \quad
 \frac{1}{n} \log W_n(g,\x) < G(g) + \frac{\eps}{4}.
\]
Therefore, for all $n\ge n_0$ and all $\x$
\[
\frac{1}{n} \log W_0 + \frac{1}{n} \log W_n(f,\x) - \frac{1}{n} \log V_0 - \frac{1}{n} \log W_n(g,\x) > \frac{\eps}{4}>0.
\]
This implies that 
\[M_n(f,g) = \PP\Bigg(\frac{1}{n} \log W_0 + \frac{1}{n} \log W_n(f)> \frac{1}{n} \log V_0 + \frac{1}{n} \log W_n(g)\Bigg) = 1 \quad \text{ for all $n\ge n_0$. }\]
Therefore, $M_{\infty}(f,g) = 1$. Similarly we obtain $M_{\infty}(f,g) = 0$ if $G(f)<G(g)$. 
\end{proof}

\begin{rem}
For the case $G(f) = G(g)$ there are only two cases, $g=f$ or $g=\hat f$. If $g=f$ we have $M_{\infty}(f,f) = \frac{1}{2}$. If $g=\hat f$ we do not know the value of $M_{\infty}(f,\hat f)$. 
\end{rem}
\noindent
See Fig.~\ref{fig:AM} for a graphical illustration of the convergence. \\

\noindent
\hspace*{3.5cm}{\bf\large A} \hspace*{4cm} {\bf\large B}
\begin{figure}[h]
\centering\includegraphics[width=4.5cm]{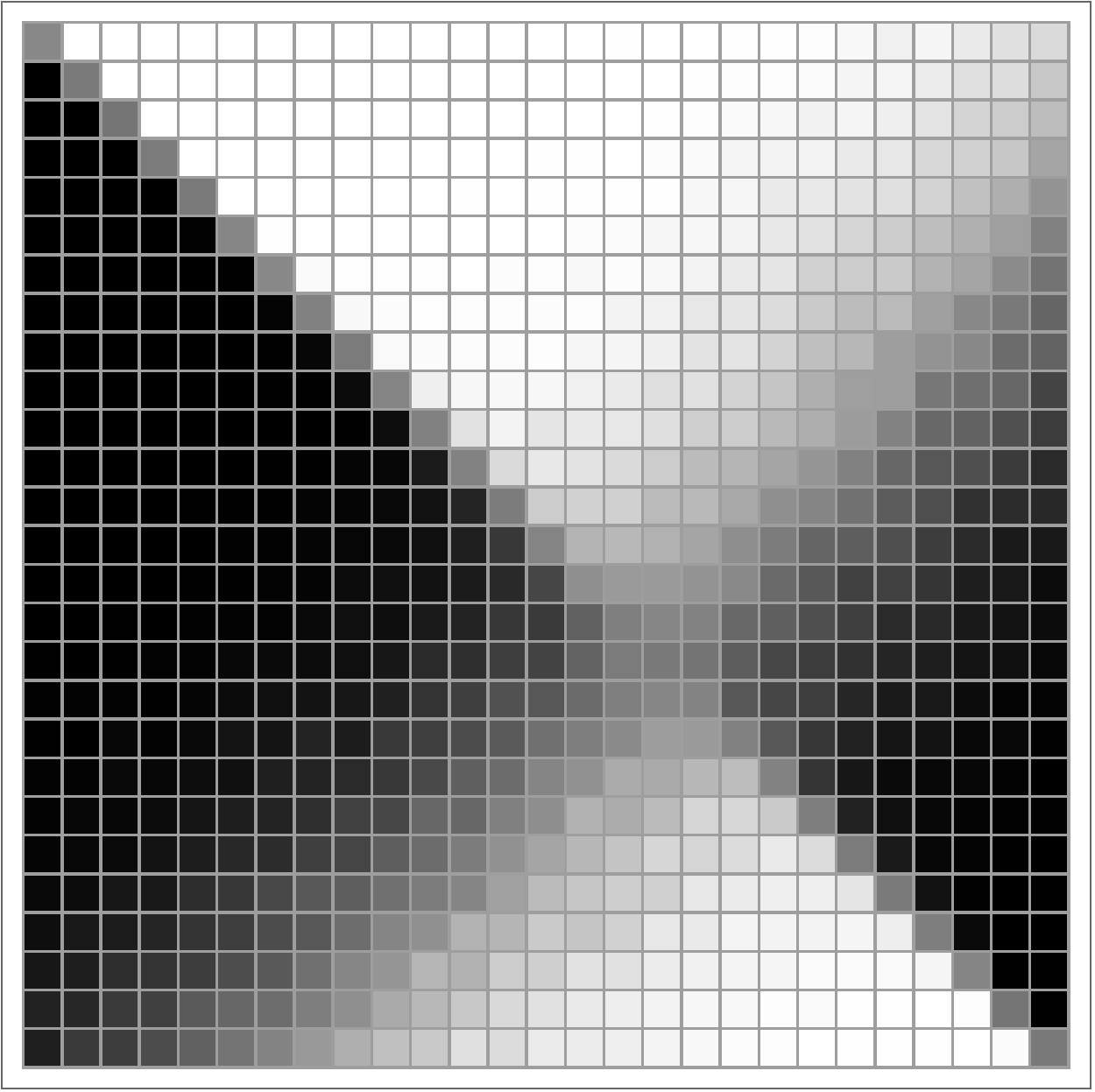}
\centering\includegraphics[width=4.5cm]{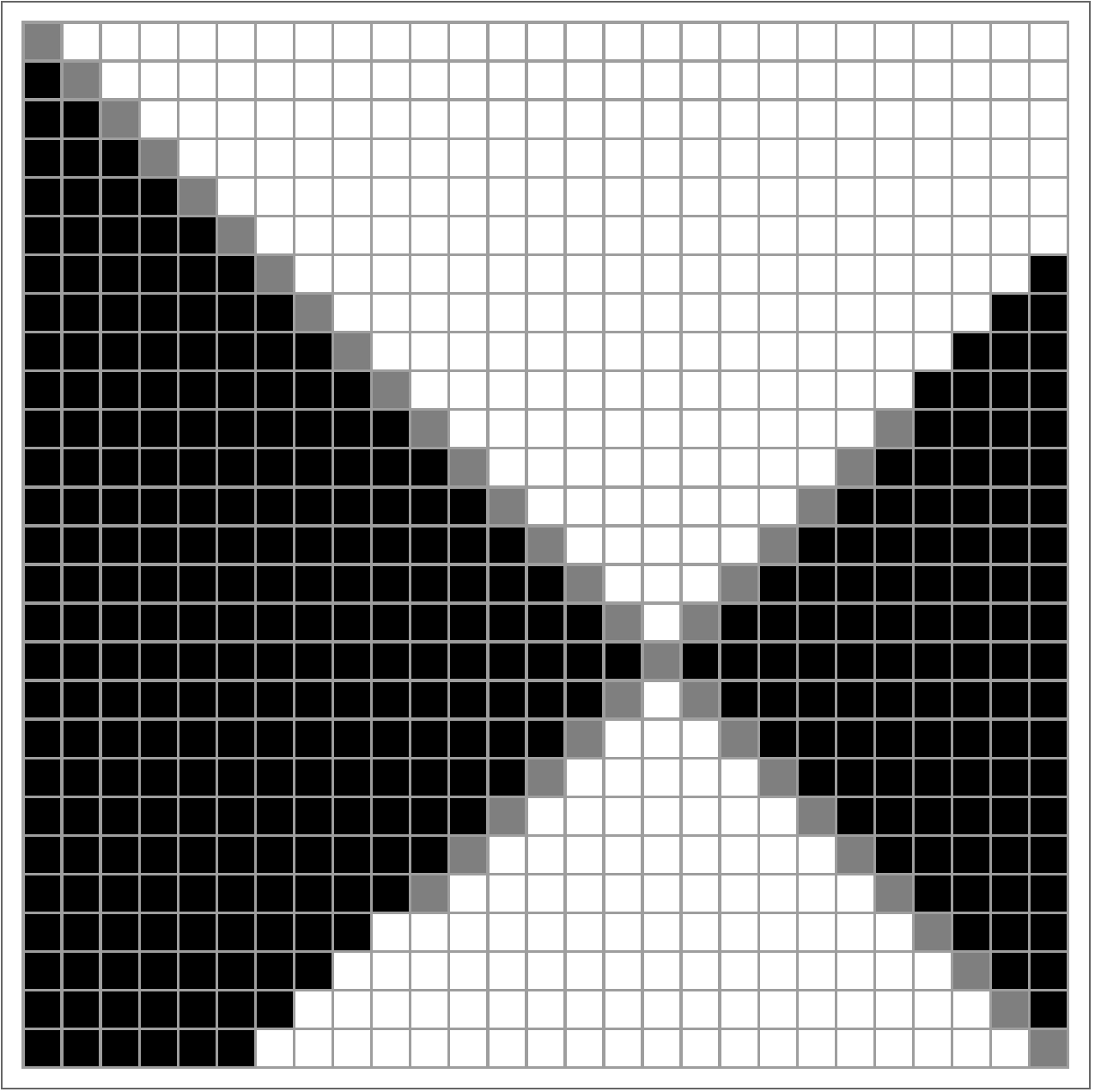}
\caption{{\small A payoff matrix resulting from a very large time horizon $n$ (for primary lineage). With increasing horizon $n$ values become highly contrasted with off-diagonal cells near $0$ or $1$ (excepting those corresponding to strategies that result in equal asymptotic growth rates, here on the ridge perpendicular to the diagonal). | B: A payoff matrix with entries based on pairwise differences in $E[\log W]$, reflecting standard Kelly asymptotic growth-optimality.}}
\label{fig:AM}
\end{figure}

\section{Nash equilibrium in population size $N$}\label{app:N}
In this section, we show that the Nash solution in population size $N$, denoted by $f^*_N$, will be the strategy closest to Kelly under the finite resolution regime, and such that it converges asymptotically with $N$ to the Kelly strategy. Denote by $f^*_N$  the closest element to $f^{Kelly}$ in $I_N:=\{0,\frac{1}{N},\ldots, 1\}$, i.e., $f^*_N = \argmin_{f\in I_N}|f-f^{Kelly}|$. We show that $(f^*_N,f^*_N)$ is the Nash solution for the game with strategies defined only on $I_N$. Due to the definition of $f^*_N$, we see that $|f^*_N - f^{Kelly}| \le \frac{1}{N} \to 0$ as $N\to \infty$. To this end, we show that 
$M_n(f^*_N, f^*_N) \ge M_n(f, f^*_N)$ for all $f \in I_N$. Indeed, we have already from \eqref{eq:equalhalf} that $M_n(f^*_N, f^*_N) = \frac{1}{2}$. Moreover,  we have  $p \log \overline{o}_1(f) + (1-p) \log \overline{o}_2(f) < p \log \overline{o}_1(f^*_N) + (1-p) \log \overline{o}_2(f^*_N)$ for all $f\in I_N \setminus \{f^*_N\}$. Therefore there exists $\eps>0$ such that 
\[p \log \frac{\overline{o}_1(f)}{\overline{o}_1(f^*_N)} + (1-p) \log \frac{\overline{o}_2(f)}{\overline{o}_2(f^*_N)} < -\eps\quad \forall f\in I_N \setminus \{f^*_N\}.
\]
Thus, for every $f\in I_N \setminus \{f^*_N\}$ we have $\sum_{s=0}^n \al(s) P(s) < -n\eps$ where $\al(s) := s \log \frac{\overline{o}_1(f)}{\overline{o}_1(f^*_N)} + (n-s) \log \frac{\overline{o}_2(f)}{\overline{o}_2(f^*_N)}$. We assume that $\log W_0$ and $\log V_0$ have the same distribution with $\supp \log W_0 \supset \{\al(0),\ldots,\al(n)\}$ and $|\supp \log W_0| = |\supp \log V_0|  = r > 2n \eps$. Denote by $A_1=\{s: \al(s)<0\}$, $A_2=\{s: \al(s)\ge 0\}$ and $\de = \frac{\frac{1}{2} - \frac{n\eps}{r}}{\frac{3}{4}- \frac{n\eps}{r}} \in (0,\frac{1}{2})$. We have 
\[
\begin{split}
M_n(f,f^*_N) &= \sum_{s=0}^n \PP(\al(s) + \log W_0 > \log V_0) P(s)=\sum_{s\in A_1} \frac{\frac{1}{2}(r+\al(s))^2}{r^2} P(s) + \sum_{s\in A_2} \Bigg(1-\frac{\frac{1}{2}(r-\al(s))^2}{r^2}\Bigg) P(s)\\
&= \frac{1}{2}+\sum_{s\in A_1} \frac{\al(s)}{r}\Big(1+\frac{\al(s)}{2r}\Big)P(s)+\sum_{s\in A_2} \frac{\al(s)}{r}\Big(1-\frac{\al(s)}{2r}\Big)P(s)\\
&= \frac{1}{2}+\sum_{s\in A_1} \frac{\al(s)}{r}\Big(\de+\frac{\al(s)}{2r}\Big)P(s)+\sum_{s\in A_2} \frac{\al(s)}{r}\Big(\de-\frac{\al(s)}{2r}\Big)P(s) + (1-\de) \sum_{s=0}^n  \frac{\al(s)}{r}P(s).
\end{split}
\]
Note that $\frac{\al(s)}{r}\in [-1,0]$ for $s\in A_1$ and $\frac{\al(s)}{r}\in [0,1]$ for $s\in A_2$. Moreover $x(\de + x/2) \le \frac{1}{2}-\de < \frac{1}{2} - \frac{3}{4} \de $ for $x\in [-1,0]$; $x(\de - x/2) \le \frac{\de^2}{2} <  \frac{1}{2} - \frac{3}{4} \de$ for $x\in [-1,0]$. Therefore, for every $f\in I_N \setminus \{f^*_N\}$ we have
\[
\begin{split}
M_n(f,f^*_N) &< \frac{1}{2}+\Big(\frac{1}{2} - \frac{3}{4} \de\Big) \sum_{s=0}^n P(s) + (1-\de) \sum_{s=0}^n  \frac{\al(s)}{r}P(s)<  \frac{1}{2}+\Big(\frac{1}{2} - \frac{3}{4} \de\Big) + (1-\de) \frac{-n\eps}{r} = \frac{1}{2}.
\end{split}
\]

\section{Nash equilibrium in nonstationary environments}\label{app:O}
\begin{prop}\label{prop:Nashnonstationary}
We consider also a game with payoff of players are
\[
M_n(f,g) := \PP(W_n(f) W_0 > W_n(g) V_0), 
\]
where $W_0, V_0$ have the same distribution.
Then, in this game, $(f^{Kelly},f^{Kelly})$ is the unique strict Nash equilibrium.
\end{prop}
\begin{proof}
We note that in the non-stationary case we have
\[
\begin{split}
M_n(f,g) &=  \sum_{H=0}^n \PP(\al_1^H \al_2^{n-H} W_0 > V_0) P(H),\\
\end{split}
\]
where $H\sim GB\big(n,\{p_1,\ldots, p_n\}\sim Beta(\al,\beta)\big)$ is a generalized binomial distribution. Therefore the proof is similar to the proof in Proposition~\ref{prop:3} and is omitted.
\end{proof}

\section{Limit of the extinction rate}\label{app:P}

\begin{prop} Denote by 
\[
Q_{n,d}(f) := \PP(W_0 W_1(f) > d, \ldots,  W_0 W_n(f)> d).
\]
the probability that the extinction does not occur until time $n$ and $P_{n,d}(f) = 1-Q_{n,d}(f)$ the probability of extinction until time $n$ (also see Fig.~\ref{fig:6}). We prove that
\[
\lim_{n\to \infty}P_{n,d}(f) =  \begin{cases}
0,& \text{ if } \overline{o}_1(f), \overline{o}_2(f)>1\\
1,& \text{ if } \overline{o}_1(f), \overline{o}_2(f)<1\\
c_d(f) \in [0,1],& \text{ else.}
\end{cases}
\]
\end{prop}
\begin{proof}
For the sake of simplicity, we denote by 
\[\beta_{n,d}(x_1,\ldots,x_n):= \frac{d}{\overline{o}_1(f)^{x_1}\overline{o}_2(f)^{1-x_1}} \vee \cdots \vee \frac{d}{\overline{o}_1(f)^{x_1+\cdots+x_n}\overline{o}_2(f)^{n-x_1-\cdots-x_n}}.
\]
Then we rewrite the formula 
\[
\begin{split}
Q_{n,d}(f) &= \sum_{x_1,\ldots,x_n} P(x_1,\ldots,x_n) \PP\big(W_0 > \beta_{n,d}(x_1,\ldots,x_n) \big).
\end{split}
\]
\begin{enumerate}
\item[(i)] If $\overline{o}_1(f), \overline{o}_2(f)>1$: we have $\beta_{n,d}(x_1,\ldots,x_{n-1},1) = \beta_{n,d}(x_1,\ldots,x_{n-1},0) = \beta_{n-1,d}(x_1,\ldots,x_{n-1})$, therefore  $Q_{n,d}(f) = Q_{n-1,d}(f) =\cdots = Q_{0,d}= \PP(W_0>d) = 1$ for all $n$. Therefore $\lim_{n\to \infty}P_{n,d}(f) =  0$.
\item[(ii)] If $\overline{o}_1(f), \overline{o}_2(f)<1$: we have $\beta_{n,d}(x_1,\ldots,x_{n}) = \frac{d}{\overline{o}_1(f)^{x_1+\cdots+x_n}\overline{o}_2(f)^{n-x_1-\cdots-x_n}}$ which approaches infinity with $n$. Therefore for $n$ large enough, $Q_{n,d} = 0$. This implies $\lim_{n\to \infty}P_{n,d}(f) =  1$.
\item[(iii)] If $\overline{o}_1(f) > 1>  \overline{o}_2(f)$: we have $\beta_{n,d}(x_1,\ldots,x_{n-1},1) = \beta_{n-1,d}(x_1,\ldots,x_{n-1})$. Note that $Q_{n,d}$ is decreasing and bounded below by $0$, therefore there exists the limit of $Q_{n,d}(f)$ which implies the limit of $P_{n,d}(f)$.
\end{enumerate}
\end{proof}
\begin{rem}
$c_d(f)$ is increasing with $d$ (see Fig.~\ref{fig:6}). 
\end{rem}
\begin{proof}
If $d_1>d_2$ then $\beta_{n,d_1}(x_1,\ldots,x_n) > \beta_{n,d_2}(x_1,\ldots,x_n)$ therefore $Q_{n,d_1}(f) < Q_{n,d_2}(f)$ which implies that $c_{d_1}(f) \ge c_{d_2}(f)$. 
\end{proof}

\end{appendices}

\newpage
\bibliographystyle{apalike}

\begin{thebibliography}{}

\bibitem[Ashby et~al., 2017]{ashby17}
Ashby, B., Watkins, E., Louren{\c c}o, J., Gupta, S., and Foster, K.~R. (2017).
\newblock Competing species leave many potential niches unfilled.
\newblock {\em Nature Ecology \& Evolution}, 1:1495--1501.

\bibitem[Aucamp, 1977]{aucamp77}
Aucamp, D. (1977).
\newblock An investment strategy with overshoot rebates which minimizes the
  time to attain a specified goal.
\newblock {\em Management Science}, 23(11).

\bibitem[Bell and Cover, 1988]{bc88}
Bell, R. and Cover, T.~M. (1988).
\newblock Game-theoretic optimal portfolios.
\newblock {\em Management Science}, 34(6):724--733.

\bibitem[Bell and Cover, 1980]{bc80}
Bell, R.~M. and Cover, T.~M. (1980).
\newblock Competitive optimality of logarithmic investment.
\newblock {\em Mathematics of Operations Research}, 5(2):161--166.

\bibitem[Bergstrom, 2014]{bergstrom14}
Bergstrom, T.~C. (2014).
\newblock On the evolution of hoarding, risk-taking, and wealth distribution in
  nonhuman and human populations.
\newblock {\em Proceedings of the National Academy of Sciences}, 111(Supplement
  3):10860--10867.

\bibitem[Breiman, 1961]{breiman61}
Breiman, L. (1961).
\newblock Optimal gambling systems for favorable games.
\newblock {\em Proceedings of the 4th Berkeley Symposium on Mathematical
  Statistics and Probability}, 1:63--68.

\bibitem[Busseti et~al., 2016]{busseti16}
Busseti, E., Ryu, E.-K., and Boyd, S. (2016).
\newblock Risk constrained kelly gambling.
\newblock {\em The Journal of Investing}, 25(3):118--134.

\bibitem[Cohen, 1966]{cohen66}
Cohen, D. (1966).
\newblock Optimizing reproduction in a randomly varying environment.
\newblock {\em J Theor Biol.}, 12(1):119--29.

\bibitem[Cover and Thomas, 2006]{ct06}
Cover, T.~M. and Thomas, J.~A. (2006).
\newblock {\em Elements of information theory}.
\newblock Wiley-Interscience [John Wiley \& Sons], Hoboken, NJ, second edition.

\bibitem[Donaldson-Matasci et~al., 2010]{donaldson10}
Donaldson-Matasci, M.~C., Bergstrom, C.~T., and Lachmann, M. (2010).
\newblock The fitness value of information.
\newblock {\em Oikos (Copenhagen, Denmark)}, 119(2):219--230.

\bibitem[Garivaltis, 2018]{garivaltis18}
Garivaltis, A. (2018).
\newblock Game-theoretic optimal portfolios in continuous time.
\newblock {\em Econ Theory Bull}, pages 1--9.

\bibitem[Gremer and Venable, 2014]{gv14}
Gremer, J.~R. and Venable, D.~L. (2014).
\newblock Bet hedging in desert winter annual plants: optimal germination
  strategies in a variable environment.
\newblock {\em Ecology Letters}, 17(3):380--387.

\bibitem[Hakansson, 1971]{hakansson71}
Hakansson, N. (1971).
\newblock Capital growth and the mean-variance approach to portfolio selection.
\newblock {\em The Journal of Financial and Quantitative Analysis},
  6(1):517--557.

\bibitem[Hopper, 2018]{hopper18}
Hopper, K.~R. (2018).
\newblock Bet hedging in evolutionary ecology with an emphasis on insects.
\newblock In {\em Reference Module in Life Sciences}. Elsevier.

\bibitem[Kardaras and Platen, 2010]{kp10}
Kardaras, C. and Platen, E. (2010).
\newblock Minimizing the expected market time to reach a certain wealth level.
\newblock {\em SIAM Journal on Financial Mathematics}, 1(1):16--29.

\bibitem[Kelly, 1956]{Kelly56}
Kelly, Jr., J.~L. (1956).
\newblock A new interpretation of information rate.
\newblock {\em Bell. System Tech. J.}, 35:917--926.

\bibitem[King and Masel, 2007]{km07}
King, O. and Masel, J. (2007).
\newblock The evolution of bet-hedging adaptations to rare scenarios.
\newblock {\em Theor. Popul. Biol.}, 72(4):560--75.

\bibitem[Kussell et~al., 2005]{kussell05}
Kussell, E., Kishony, R., Balaban, N., and Leibler, S. (2005).
\newblock Bacterial persistence: a model of survival in changing environments.
\newblock {\em Genetics.}, 169(4):1807--14.

\bibitem[Lande, 2007]{lande07}
Lande, R. (2007).
\newblock Expected relative fitness and the adaptive topography of fluctuating
  selection.
\newblock {\em Evolution}, 61:1835--1846.

\bibitem[Li et~al., 2017]{xiangyi17}
Li, X.-Y., Lehtonen, J., and Kokko, H. (2017).
\newblock {\em Sexual Reproduction as Bet-Hedging}, pages 217--234.
\newblock Springer International Publishing, Cham.

\bibitem[Libby and Ratcliff, 2019]{lr19}
Libby, E. and Ratcliff, W.~C. (2019).
\newblock Shortsighted evolution constrains the efficacy of long-term bet
  hedging.
\newblock {\em The American Naturalist}, 193(3):409--423.
\newblock PMID: 30794447.

\bibitem[Lo et~al., 2017]{lo17}
Lo, A., Orr, H., and Zhang, R. (2017).
\newblock The growth of relative wealth and the kelly criterion.
\newblock {\em Journal of Bioeconomics}, 20(1):49--67.

\bibitem[MacLean et~al., 2011]{maclean10}
MacLean, L.~C., Thorp, E.~O., and Ziemba, W.~T. (2011).
\newblock {\em Good and Bad Properties of the Kelly Criterion}, chapter~39,
  pages 563--572.
\newblock World Scientific Handbook in Financial Economics Series.

\bibitem[Markowitz, 2006]{markowitz06}
Markowitz, H. (2006).
\newblock {\em Samuelson and Investment for the Long Run.}, pages 252--261.
\newblock Oxford University Press.
\newblock In Samuelsonian Economics and the Twenty-First Century.

\bibitem[Morgan, 2015]{morgan15}
Morgan, D. (2015).
\newblock An alternative mathematical interpretation and generalization of the
  capital growth criterion.
\newblock {\em Journal of Finance and Investment Analysis}, 4(4):6.

\bibitem[Nash, 1951]{nash51}
Nash, J. (1951).
\newblock Non-cooperative games.
\newblock {\em Annals of Mathematics}, 54(2):286--295.

\bibitem[Okasha, 2018]{okasha18}
Okasha, S. (2018).
\newblock {\em Agents and goals in evolution}.
\newblock Oxford University Press.

\bibitem[Olofsson et~al., 2009]{olofsson09}
Olofsson, H., Ripa, J., and Jonz{\'e}n, N. (2009).
\newblock Bet-hedging as an evolutionary game: the trade-off between egg size
  and number.
\newblock {\em Proceedings. Biological sciences}, 276(1669):2963--2969.

\bibitem[Orr, 2017]{orr17}
Orr, H. (2017).
\newblock Evolution, finance, and the population genetics of relative wealth.
\newblock {\em Journal of Bioeconomics}, 20(1).

\bibitem[Proulx and Day, 2001]{pd01}
Proulx, S.~R. and Day, T. (2001).
\newblock What can invasion analyses tell us about evolution under
  stochasticity?
\newblock {\em Selection}, 2(1-2):1--15.

\bibitem[Ram et~al., 2018]{Ram18}
Ram, Y., Liberman, U., and Feldman, M.~W. (2018).
\newblock Evolution of vertical and oblique transmission under fluctuating
  selection.
\newblock {\em Proceedings of the National Academy of Sciences},
  115(6):E1174--E1183.

\bibitem[Reznick et~al., 2002]{reznick02}
Reznick, D., Bryant, M.~J., and Bashey, F. (2002).
\newblock r- and k-selection revisited: The role of population regulation in
  life-history evolution.
\newblock {\em Ecology}, 83(6):1509--1520.

\bibitem[Rivoire and Leibler, 2011]{rl11}
Rivoire, O. and Leibler, S. (2011).
\newblock The value of information for populations in varying environments.
\newblock {\em Journal of Statistical Physics}, 142(6):1124--1166.

\bibitem[Rubin and Doebeli, 2017]{rd17}
Rubin, I.~N. and Doebeli, M. (2017).
\newblock Rethinking the evolution of specialization: A model for the evolution
  of phenotypic heterogeneity.
\newblock {\em Journal of Theoretical Biology}, 435:248 -- 264.

\bibitem[Rujeerapaiboon et~al., 2015]{rujeerapaiboon15}
Rujeerapaiboon, N., Kuhn, D., and Wiesemann, W. (2015).
\newblock Robust growth-optimal portfolios.
\newblock {\em Management Science}, 62(7):2090--2109.

\bibitem[{Rujeerapaiboon} et~al., 2018]{rujeerapaiboon18}
{Rujeerapaiboon}, N., {Ross Barmish}, B., and {Kuhn}, D. (2018).
\newblock On risk reduction in kelly betting using the conservative expected
  value.
\newblock In {\em 2018 IEEE Conference on Decision and Control (CDC)}, pages
  5801--5806.

\bibitem[Seger and Brockmann, 1987]{sb87}
Seger, J. and Brockmann, H.~J. (1987).
\newblock {\em What is bet-hedging?}, pages 182--211.
\newblock Oxford, UK: Oxford University Press.
\newblock Oxford surveys in evolutionary biology.

\bibitem[{Shannon}, 1956]{shannon56}
{Shannon}, C. (1956).
\newblock The bandwagon (edtl.).
\newblock {\em IRE Transactions on Information Theory}, 2(1):3--3.

\bibitem[Simons and Johnston, 2003]{sj03}
Simons, A. and Johnston, M. (2003).
\newblock Suboptimal timing of reproduction in lobelia inflata may be a
  conservative bet-hedging strategy.
\newblock {\em J. Evol. Biol.}, 16:233--243.

\bibitem[Smith and Price, 1973]{sp73}
Smith, J. and Price, G. (1973).
\newblock The logic of animal conflict.
\newblock {\em Nature 246 15--18}, 246:15--18.

\bibitem[Stollmeier and Nagler, 2018]{sn18}
Stollmeier, F. and Nagler, J. (2018).
\newblock Unfair and anomalous evolutionary dynamics from fluctuating payoffs.
\newblock {\em Physical Review Letters}, 120:058101.

\bibitem[Villa-Martin et al., 2019]{martine19}
Villa~Martin, P., Munoz, M.~A., and Pigolotti, S. (2019).
\newblock Bet-hedging strategies in expanding populations.
\newblock {\em PLOS Computational Biology}, 15(4):1--17.

\bibitem[Vince and Zhu, 2013]{vz13}
Vince, R. and Zhu, Q. (2013).
\newblock Inflection point significance for the investment size.
\newblock {\em Available at SSRN: https://ssrn.com/abstract=2230874}.

\bibitem[Wolf et~al., 2005]{wolf05}
Wolf, D.~M., Vazirani, V.~V., and Arkin, A.~P. (2005).
\newblock Diversity in times of adversity: probabilistic strategies in
  microbial survival games.
\newblock {\em Journal of Theoretical Biology}, 234(2):227 -- 253.

\bibitem[Yoshimura and Jansen, 1996]{yj96}
Yoshimura, J. and Jansen, V. A.~A. (1996).
\newblock Evolution and population dynamics in stochastic environments.
\newblock {\em Researches on Population Ecology}, 38(2):165--182.

\bibitem[Yoshimura et~al., 2009]{yoshimura09}
Yoshimura, J., Tanaka, Y., Togashi, T., Iwata, S., and ichi Tainaka, K. (2009).
\newblock Mathematical equivalence of geometric mean fitness with probabilistic
  optimization under environmental uncertainty.
\newblock {\em Ecological Modelling}, 220(20):2611 -- 2617.

\end{thebibliography}

\end{document}